\documentclass[graybox]{svmult}

\usepackage{mathptmx}       
\usepackage{helvet}         
\usepackage{courier}        
\usepackage{type1cm}        
\usepackage{makeidx}         
\usepackage{graphicx}        
\usepackage{multicol}        
\usepackage[bottom]{footmisc}

\makeindex             

\usepackage{breakurl}

\usepackage{tikzfig}

\newcommand{\dotcounit}[1]{%
\,\begin{tikzpicture}[dotpic,yshift=-1mm]
\node [#1] (a) at (0,0.25) {}; 
\draw [medium diredge] (0,-0.2)--(a);
\end{tikzpicture}\,}
\newcommand{\dotunit}[1]{%
\,\begin{tikzpicture}[dotpic,yshift=1.5mm]
\node [#1] (a) at (0,-0.25) {}; 
\draw [medium diredge] (a)--(0,0.2);
\end{tikzpicture}\,}
\newcommand{\dotcomult}[1]{%
\,\begin{tikzpicture}[dotpic,yshift=0.5mm]
  \node [#1] (a) {};
  \draw [medium diredge] (-90:0.55)--(a);
  \draw [medium diredge] (a) -- (45:0.6);
  \draw [medium diredge] (a) -- (135:0.6);
\end{tikzpicture}\,}
\newcommand{\dotmult}[1]{%
\,\begin{tikzpicture}[dotpic]
  \node [#1] (a) {};
  \draw [medium diredge] (a) -- (90:0.55);
  \draw [medium diredge] (a) (-45:0.6) -- (a);
  \draw [medium diredge] (a) (-135:0.6) -- (a);
\end{tikzpicture}\,}

\newcommand{\dotonly}[1]{%
\,\begin{tikzpicture}[dotpic,yshift=-0.3mm]
\node [#1] (a) at (0,0) {};
\end{tikzpicture}\,}
\newcommand{\smalldotonly}[1]{%
\,\begin{tikzpicture}[dotpic,yshift=-0.15mm]
\node [#1] (a) at (0,0) {};
\end{tikzpicture}\,}
\newcommand{\smallblackdot}{\smalldotonly{smalldot}}
\newcommand{\unit}{\dotunit{dot}}

\newcommand{\blackobs}{\ensuremath{\mathcal O_{\!\smallblackdot}}\xspace}

\newcommand{\whitedot}{\dotonly{white dot}}
\newcommand{\smallwhitedot}{\smalldotonly{smallwhitedot}}
\newcommand{\whiteunit}{\dotunit{white dot}}
\newcommand{\whitecounit}{\dotcounit{white dot}}
\newcommand{\whitemult}{\dotmult{white dot}}
\newcommand{\whitecomult}{\dotcomult{white dot}}

\newcommand{\whiteobs}{\ensuremath{\mathcal O_{\!\smallwhitedot}}\xspace}

\newcommand{\graydot}{\dotonly{gray dot}}
\newcommand{\smallgraydot}{\smalldotonly{smallgraydot}}
\newcommand{\grayunit}{\dotunit{gray dot}}

\newcommand{\graymult}{\dotmult{gray dot}}

\newcommand{\grayobs}{\ensuremath{\mathcal O_{\!\smallgraydot}}\xspace}

\newcommand{\graytranspose}{\ensuremath{{\,\graydot\!\textrm{\rm\,T}}}}
\newcommand{\whitetranspose}{\ensuremath{{\,\whitedot\!\textrm{\rm\,T}}}}

\newcommand{\whiteconjugate}{\circledast}

\tikzstyle{env}=[copoint,regular polygon rotate=0,minimum width=0.2cm, fill=black]

\tikzstyle{every picture}=[baseline=-0.25em]
\tikzstyle{dotpic}=[scale=0.5]
\tikzstyle{diredges}=[every to/.style={diredge}]
\tikzstyle{dot graph}=[shorten <=-0.1mm,shorten >=-0.1mm,scale=0.6]
\tikzstyle{plot point}=[circle,fill=black,minimum width=2mm,inner sep=0]

\tikzstyle{braceedge}=[decorate,decoration={brace,amplitude=2mm,raise=-1mm}]
\tikzstyle{small braceedge}=[decorate,decoration={brace,amplitude=1mm,raise=-1mm}]
\tikzstyle{left hook arrow}=[left hook-latex]
\tikzstyle{right hook arrow}=[right hook-latex]

\tikzstyle{scalar}=[diamond,draw=black,inner sep=0.5pt,font=\small]
\tikzstyle{black dot}=[inner sep=0.7mm,minimum width=0pt,minimum height=0pt,fill=black,draw=black,shape=circle]

\tikzstyle{dot}=[black dot]
\tikzstyle{smalldot}=[inner sep=0.4mm,minimum width=0pt,minimum height=0pt,fill=black,draw=black,shape=circle]
\tikzstyle{white dot}=[dot,fill=white]
\tikzstyle{antipode}=[white dot,inner sep=0.3mm,font=\footnotesize]
\tikzstyle{smallwhitedot}=[smalldot,fill=white]
\tikzstyle{alt white dot}=[white dot,label={[xshift=3.07mm,yshift=-0.05mm,font=\footnotesize]left:$*$}]
\tikzstyle{gray dot}=[dot,fill=gray!40!white]
\tikzstyle{smallgraydot}=[smalldot,fill=gray!40!white]
\tikzstyle{box vertex}=[draw=black,rectangle]
\tikzstyle{small box}=[box vertex,fill=white]
\tikzstyle{whitebg}=[fill=white,inner sep=2pt]
\tikzstyle{graph state vertex}=[sg vertex,fill=black]

\tikzstyle{wide copoint}=[fill=white,draw=black,shape=isosceles triangle,shape border rotate=90,isosceles triangle stretches=true,inner sep=1pt,minimum width=1.5cm,minimum height=5mm]
\tikzstyle{wide point}=[fill=white,draw=black,shape=isosceles triangle,shape border rotate=-90,isosceles triangle stretches=true,inner sep=1pt,minimum width=1.5cm,minimum height=4mm]
\tikzstyle{very wide copoint}=[fill=white,draw=black,shape=isosceles triangle,shape border rotate=-90,isosceles triangle stretches=true,inner sep=1pt,minimum width=2.5cm,minimum height=4mm]
\tikzstyle{very wide empty copoint}=[draw=black,shape=isosceles triangle,shape border rotate=-90,isosceles triangle stretches=true,inner sep=1pt,minimum width=2.5cm,minimum height=4mm]
\tikzstyle{symm}=[ultra thick,shorten <=-1mm,shorten >=-1mm]

\tikzstyle{andgate}=[fill=white,draw=black,shape=isosceles
triangle,shape border rotate=90,isosceles triangle
stretches=true,inner sep=0pt,minimum width=6mm,minimum height=4mm]
\tikzstyle{orgate}=[fill=black,draw=black,shape=isosceles
triangle,shape border rotate=90,isosceles triangle
stretches=true,inner sep=0pt,minimum width=6mm,minimum height=4mm]

 \tikzstyle{notgate}=[box vertex,fill=gray,inner sep=0pt,minimum width=3mm,minimum height=3mm]
\tikzstyle{fangate}=[black dot]

\tikzstyle{square box}=[rectangle,fill=white,draw=black,minimum height=5mm,minimum width=5mm,font=\small]
\tikzstyle{square gray box}=[rectangle,fill=gray!30,draw=black,minimum height=6mm,minimum width=6mm]
\tikzstyle{copoint}=[regular polygon,regular polygon sides=3,draw=black,scale=0.75,inner sep=-0.5pt,minimum width=7mm,fill=white]
\tikzstyle{point}=[regular polygon,regular polygon sides=3,draw=black,scale=0.75,inner sep=-0.5pt,minimum width=7mm,fill=white,regular polygon rotate=180]
\tikzstyle{gray point}=[point,fill=gray!40!white]
\tikzstyle{gray copoint}=[copoint,fill=gray!40!white]

\newcommand{\edgearrow}{{\arrow[black]{>}}}
\newcommand{\edgetick}{{\arrow[black,scale=0.7,very thick]{|}}}

\tikzstyle{diredge}=[->]
\tikzstyle{rdiredge}=[<-]
\tikzstyle{medium diredge}=[->]

\tikzstyle{short diredge}=[->]
\tikzstyle{halfedge}=[-)]
\tikzstyle{other halfedge}=[(-]
\tikzstyle{freeedge}=[(-)]
\tikzstyle{white edge}=[line width=5pt,white]
\tikzstyle{tick}=[postaction=decorate,decoration={markings, mark=at position 0.5 with \edgetick}]
\tikzstyle{small map edge}=[|-latex, gray!60!blue, shorten <=0.9mm, shorten >=0.5mm]
\tikzstyle{thick dashed edge}=[very thick,dashed,gray!40]
\tikzstyle{dashed edge}=[dashed,gray!40]
\tikzstyle{map edge}=[|-latex,very thick, gray!40, shorten <=1mm, shorten >=0.5mm]
\tikzstyle{tickedge}=[postaction=decorate,
  decoration={markings, mark=at position 0.5 with \edgetick}]
\tikzstyle{dirtickedge}=[postaction=decorate,
  decoration={markings, mark=at position 0.5 with \edgetick},
  decoration={markings, mark=at position 0.85 with \edgearrow}]
\tikzstyle{dirdoubletickedge}=[postaction=decorate,
  decoration={markings, mark=at position 0.4 with \edgetick},
  decoration={markings, mark=at position 0.6 with \edgetick},
  decoration={markings, mark=at position 0.85 with \edgearrow}]

\makeatletter
\newcommand{\boxshape}[3]{%
\pgfdeclareshape{#1}{
\inheritsavedanchors[from=rectangle] 
\inheritanchorborder[from=rectangle]
\inheritanchor[from=rectangle]{center}
\inheritanchor[from=rectangle]{north}
\inheritanchor[from=rectangle]{south}
\inheritanchor[from=rectangle]{west}
\inheritanchor[from=rectangle]{east}
\backgroundpath{
\southwest \pgf@xa=\pgf@x \pgf@ya=\pgf@y
\northeast \pgf@xb=\pgf@x \pgf@yb=\pgf@y

\@tempdima=#2
\@tempdimb=#3

\pgfpathmoveto{\pgfpoint{\pgf@xa - 5pt + \@tempdima}{\pgf@ya}}
\pgfpathlineto{\pgfpoint{\pgf@xa - 5pt - \@tempdima}{\pgf@yb}}
\pgfpathlineto{\pgfpoint{\pgf@xb + 5pt + \@tempdimb}{\pgf@yb}}
\pgfpathlineto{\pgfpoint{\pgf@xb + 5pt - \@tempdimb}{\pgf@ya}}
\pgfpathlineto{\pgfpoint{\pgf@xa - 5pt + \@tempdima}{\pgf@ya}}
\pgfpathclose
}
}}

\boxshape{NEbox}{0pt}{4pt}
\boxshape{SEbox}{0pt}{-4pt}
\boxshape{NWbox}{4pt}{0pt}
\boxshape{SWbox}{-4pt}{0pt}
\makeatother

\tikzstyle{map}=[draw,shape=NEbox,inner sep=3pt]
\tikzstyle{mapdag}=[draw,shape=SEbox,inner sep=3pt]
\tikzstyle{maptrans}=[draw,shape=SWbox,inner sep=3pt]
\tikzstyle{mapconj}=[draw,shape=NWbox,inner sep=3pt]

\tikzstyle{probs}=[shape=semicircle,fill=gray!40!white,draw=black,shape border rotate=180,minimum width=1.2cm]

\tikzstyle{arrs}=[-latex,font=\small,auto]
\tikzstyle{arrow plain}=[arrs]
\tikzstyle{arrow dashed}=[dashed,arrs]
\tikzstyle{arrow bold}=[very thick,arrs]
\tikzstyle{arrow hide}=[draw=white!0,-]
\tikzstyle{arrow reverse}=[latex-]
\tikzstyle{cdnode}=[]

\tikzstyle{cnot}=[fill=white,shape=circle,inner sep=-1.4pt]
\tikzstyle{wire label}=[font=\tiny, auto]

\tikzstyle{gray wide copoint}=[fill=gray,draw=black,shape=isosceles triangle,shape border rotate=90,isosceles triangle stretches=true,inner sep=1pt,minimum width=1.5cm,minimum height=5mm]

\usepackage{tikzcircuits}

\newcommand{\inline}[1]{
  \raisebox{0.5ex}{\;#1\;}
}

\usetikzlibrary{fit}

\newcommand{\circCX}{\inline{%
\begin{tikzpicture}[circuit]
  \node  (0) at (-1, 1) {};
  \node  (1) at (1, 1) {};
  \circcnot{2}{-1, 0}{3}{1, 0}
  \node  (4) at (-1, -1) {};
  \node  (5) at (1, -1) {};
  \draw  (1) to (3);
  \draw  (0) to (2);
  \draw  (3) to (5);
  \draw  (2) to (4);
\end{tikzpicture}
}}

\usepackage{amssymb}
\usepackage{amsmath}
\usepackage{latexsym}
\usepackage{color}
\usepackage{verbatim}
\usepackage{graphicx}
\usepackage{xspace}
\usepackage{todonotes}
\usepackage{stmaryrd}
\usepackage{mathtools}

\newcommand{\whitemu}{\ensuremath{\mu_{\smallwhitedot}}\xspace}

\newcommand{\grayeta}{\ensuremath {\eta_{\smallgraydot}}\xspace}
\newcommand{\whiteeta}{\ensuremath{\eta_{\smallwhitedot}}\xspace}

\newcommand{\graydelta}{\ensuremath {\delta_{\smallgraydot}}\xspace}
\newcommand{\whitedelta}{\ensuremath{\delta_{\smallwhitedot}}\xspace}

\newcommand{\grayepsilon}{\ensuremath {\epsilon_{\smallgraydot}}\xspace}
\newcommand{\whiteepsilon}{\ensuremath{\epsilon_{\smallwhitedot}}\xspace}

\newcommand{\graymeas}{\ensuremath{m_{\!\smallgraydot}}\xspace}

\newcommand{\whitePhi}{\ensuremath{\Phi_{\!\smallwhitedot}}\xspace}
\newcommand{\grayPhi}{\ensuremath{\Phi_{\!\smallgraydot}}\xspace}

\newcommand{\whiteLambda}{\ensuremath{\Lambda_{\!\smallwhitedot}}}

\newcommand{\bra}[1]{\ensuremath{\left\langle #1 \right|}}
\newcommand{\ket}[1]{\ensuremath{\left|  #1 \right\rangle}}
\newcommand{\roundket}[1]{\ensuremath{\left|  #1 \right)}}
\newcommand{\braket}[2]{\ensuremath{\langle#1|#2\rangle}}
\newcommand{\ketbra}[2]{\ensuremath{\ket{#1}\!\bra{#2}}}

\newcommand{\innp}[2]{\braket{#1}{#2}}
\newcommand{\outp}[2]{\ketbra{#1}{#2}}

\newcommand{\CX}{\wedge X}

\newcommand{\id}[1]{\ensuremath{1_{#1}}}

\newcommand{\rw}{\Rightarrow}
\newcommand{\RW}[1]{\stackrel{#1}{\Rightarrow}}
\newcommand{\EQ}[1]{\stackrel{#1}{=}}

\newcommand{\denote}[1]{
\left\llbracket #1 \right\rrbracket}

\newcommand{\catname}[1]{\ensuremath{\mathbf{#1}}\xspace}
\newcommand{\catnamesub}[2]{\ensuremath{\mathbf{#1}_{#2}}\xspace}

\newcommand{\spek}{\catname{Spek}}
\newcommand{\stab}{\catname{Stab}}
\newcommand{\qubit}{\catname{Qubit}}
\newcommand{\fhilb}{\catname{FHilb}}
\newcommand{\fhilbd}{\catnamesub{FHilb}{D}}
\newcommand{\fhilbii}{\catnamesub{FHilb}{2}}
\newcommand{\fdhilb}{\fhilb}
\newcommand{\fdhilbd}{\fhilbd}
\newcommand{\fdhilbii}{\fhilbii}
\newcommand{\frel}{\catname{FRel}}
\newcommand{\FRel}{\frel}
\newcommand{\Frel}{\frel}
\newcommand{\freld}{\catnamesub{FRel}{D}}

\newcommand{\frelii}{\catnamesub{FRel}{2}}
\newcommand{\Frelii}{\frelii}

\newcommand{\fset}{\catname{FSet}}

\newcommand{\bool}{\catname{Bool}}
\newcommand{\symgrp}{\catname{SymGrp}}
\newcommand{\toy}{\catname{Toy}}

\newcommand{\boolcircs}{\catname{BoolCirc}}
\newcommand{\boolcirc}{\boolcircs}
\newcommand{\qcircs}{\catname{QuCirc}}
\newcommand{\qcirc}{\qcircs}

\newcommand{\CC}{\catname{C}}

\newcommand{\DD}{\catname{D}}

\newcommand{\catC}{\CC}   
\newcommand{\catD}{\DD}    

\newcommand{\catFHilb}{\fhilb}  

\DeclareMathOperator{\Tr}{Tr}

\newcommand{\boxmap}[1]{\,\tikz[dotpic]{\node[style=square box] (x) {$#1$};
\draw[diredge](0,-1)--(x);
\draw[diredge](x)--(0,1);}\,}

\newcommand{\map}[1]{\,\tikz[dotpic]{\node[style=map] (x) {$#1$};\draw(0,-1.3)--(x)--(0,1.3);}\,}

\newcommand{\mapdag}[1]{\,\tikz[dotpic]{\node[style=mapdag] (x) {$#1$};\draw(0,-1.3)--(x)--(0,1.3);}\,}

\newcommand{\graypointmap}[1]{\,\tikz[dotpic]{\node[style=gray point] (x) at (0,-0.5) {$#1$};\draw (x)--(0,1);}\,}

\newcommand{\idfig}[1]{\,\tikz[dotpic]{
\draw[diredge] (0,-1) to node[wire label]{$\scriptsize #1$} (0,1);}\;}

\newcommand{\boxpoint}[1]{\,\tikz[dotpic]{\node[style=square box] (x) {$#1$};
\draw[diredge](x)--(0,1);}\,}

\newcommand{\boxcopoint}[1]{\,\tikz[dotpic]{\node[style=square box]
    (x) at (0,0){$#1$};
\draw[diredge](0,-1)--(x);
}\,}

\tikzstyle{cdiag}=[matrix of math nodes, row sep=3em, column sep=3em, text height=1.5ex, text depth=0.25ex,inner sep=0.5em]
\tikzstyle{arrow above}=[transform canvas={yshift=0.5ex}]
\tikzstyle{arrow below}=[transform canvas={yshift=-0.5ex}]

\newcommand{\iso}{\ensuremath{\cong}}

\newcommand{\sizeof}[1]{
  \left|#1\right|}

\newcommand{\whiteplus}{\mathbin{{+_{\!\!\!{}_{\smallwhitedot}}}}}

\spnewtheorem*{exercise*}{Exercise}{\bfseries}{\rmfamily}

\def\bR{\begin{color}{red}} 
\def\bB{\begin{color}{blue}}
\def\bM{\begin{color}{magenta}}
\def\bC{\begin{color}{cyan}}
\def\bW{\begin{color}{white}}
\def\bBl{\begin{color}{black}} 
\def\bG{\begin{color}{green}}
\def\bY{\begin{color}{yellow}}
\def\e{\end{color}}

\begin{document}
\title*{Generalised Compositional Theories and Diagrammatic Reasoning}
\titlerunning{GCTs and Diagrammatic Reasoning}
\author{Bob~Coecke \and Ross Duncan \and Aleks
  Kissinger \and Quanlong Wang}
\institute{
Bob~Coecke 
\at
University of Oxford, Department of Computer Science,\\
Wolfson Building, Parks Road, Oxford OX1 3QD, UK\\
\email{coecke@cs.ox.ac.uk}
\and 
Ross Duncan 
\at
University of Strathclyde, Department of Computer and
Information Sciences,\\
  Livingston Tower, 26 Richmond Street, Glasgow G1 1XH, UK\\
\email{ross.duncan@strath.ac.uk}
\and Aleks Kissinger 
\at
University of Oxford, Department of Computer Science,\\
Wolfson Building, Parks Road, Oxford OX1 3QD, UK\\
\email{alek@cs.ox.ac.uk}
\and Quanlong Wang
\at
Beihang University, School of Mathematics and System Sciences,\\
XueYuan Road No.37, HaiDian District, Beijing, China\\
\email{qlwang@buaa.edu.cn}
}
\authorrunning{Coecke, Duncan, Kissinger, and Wang}
\maketitle
\newpage
\setcounter{minitocdepth}{2}
\dominitoc

This chapter provides an introduction to the use of \emph{diagrammatic
  language}, or perhaps more accurately, \emph{diagrammatic calculus},
in quantum information and quantum foundations.  We illustrate
the use of diagrammatic calculus in one particular case, namely the study
of \emph{complementarity} and \emph{non-locality}, two fundamental
concepts of quantum theory whose relationship we explore in later part
of this chapter.

The diagrammatic calculus that we are concerned with here is not
merely an illustrative tool, but it has both (i) a conceptual physical
backbone, which allows it to act as a foundation for diverse physical
theories, and (ii) a genuine mathematical underpinning, permitting
one to relate it to standard mathematical structures.

(i) The conceptual physical backbone concerns \emph{compositionality}.
Given two systems, there is also a composite system.  This notion of
composition is a primitive ingredient of the diagrammatic language.
Moreover, the basic elements of the diagrammatic language are
\emph{processes}, and states are identified with preparation
processes.  This paves the way for a framework of \emph{generalised
  compositional theories} (GCTs), named in analogy to generalised
probabilistic theories \cite{Barrett2007}.  The latter have recently
received much attention because one can better understand a
theory---quantum theory in particular---by studying it as merely a
member of a broader class of theories.  Notably, the study of
non-locality within this framework has provided important new insights
\cite{InfoCaus,InfoCaus2}.  Whereas generalised probabilistic theories
discard everything except the convex probabilistic structure, in
contrast, GCTs focus on composition.  This approach is informed by
techniques used in computer science, logic, and the branch of
mathematics called category theory, however its roots can be traced to
Schr\"odinger's conviction that the essential characteristic of
quantum theory is the manner in which systems compose
\cite{Schrodinger}.

(ii) On the other hand, the diagrammatic language has a well-defined
mathematical meaning, which permits any diagram to be interpreted as a
definite object in various other concrete mathematical models, for
example in Hilbert spaces.  This translation can be carried out in a
formally precise manner, so that reasoning in the diagrammatic
calculus produces true equations in the chosen model.  At the same
time, the relationship between what is provable in the calculus and
what is provable in concrete models can be described to a high degree
of precision.

We won't discuss this mathematical basis in detail here, however it
may be summarised as follows:  the diagrammatic calculus is itself a
GCT, and GCTs form a certain class of  \emph{monoidal categories},
also known as \emph{tensor categories}. 
The use of diagrammatic languages for tensors traces back
to Penrose in the early 1970's \cite{Penrose}, but was only placed on
a formal mathematical basis in the late 1980's
\cite{KellyLaplaza,JS:1993:GeoTenCal1}.  
Their use in quantum
foundations and quantum information began with an abstract (partial)
axiomatisation of Hilbert spaces in terms of these categories
\cite{AC2004}, eventually resulting in so-called \emph{quantum picturalism}
\cite{CoeckeKinder}.
Meanwhile, the diagrammatic compositional language has been adopted by
several researchers in quantum foundations \cite{Chiri,Hardy}.  The
particular developments related here been used to solve problems in
quantum foundations \cite{CES2011, Bill2} and quantum computation
\cite{CoeckeKissinger2010,DuncanPerdrix2010,Horsman}.

\section{Introduction to Quantum Picturalism}
\label{sec:intr-quant-pict}

\subsection{Theories and Diagrams}
\label{sec:theories-diagrams}

A generalised compositional theory consists of \emph{systems}, or more
accurately \emph{types of   systems}, and \emph{processes} which transform
systems.  A process $f$ which transforms systems of type $A$ into systems of
type $B$ is written $f:A\to B$.  At the highest level of generality we do not
need to give any details as to what $A$, $B$, or $f$ are: it is enough to know
that that $f$ accepts systems of type $A$ as inputs and produces systems of
type $B$ as outputs.  The important thing is how systems and processes are
combined.

Mathematically speaking, general compositional theories are
\emph{strict symmetric monoidal categories}, and a full exposition of
their properties would require a lengthy detour into category theory.
The interested reader can refer to Mac Lane's classic text
\cite{MacLane:CatsWM:1971} for a thorough treatment.  However, we can
avoid reading Mac Lane's book\footnote{We jest; reading Mac Lane's
  book is eventually unavoidable, however the paper
  \cite{BobEric2011cats} is an easy introduction to the subject of
  monoidal categories.} by adopting a diagrammatic
notation, which absorbs all of the relevant equations into the syntax.
This notation is the subject of the first section of this paper.

We will represent processes by \emph{diagrams},
consisting of boxes and wires.  The wires are labelled by systems, and
the boxes by basic processes\footnote{The term ``basic'' simply means
  a process whose internal structure is of no interest.  Typically we
  construct diagrams from some given set of basic processes.}.  Wires
join boxes at the top and bottom; the wires below correspond to the
input systems of the process, and those at the top correspond to the
output systems.  For example:
\[
\begin{array}{ccccc}
    \InputIfFileExists{w1b.tikz}{}{\input{./figures/w1b.tikz}} && \InputIfFileExists{w1a.tikz}{}{\input{./figures/w1a.tikz}} && \InputIfFileExists{w1c.tikz}{}{\input{./figures/w1c.tikz}}
    \\ \\
    f:A\to B & \qquad &
    g:A\otimes B \to B \otimes A  &  \qquad  & 
     \delta : A \to A \otimes A  
  \end{array}
\]
The same is true for the diagram as a whole:  the wires entering the
bottom of the diagram are its input systems, and those leaving from
the top are its outputs.  

Given processes $f:A\to B$ and $g:B\to C$, it seems obvious that doing $f$ then
$g$ is again a process, and we write $g\circ f:A\to C$ to denote this
process.  In other words, processes admit \emph{sequential
  combination}; we will usually call this operation
\emph{composition}.

Similarly, a pair of systems, say $A$ and $B$, can be taken together
and viewed as a single system, $A \otimes B$.  Now, given a pair of
processes $f:A\to B$ and $g:A'\to B'$, a new process is obtained by
placing them in parallel.  We denote the combined process $f\otimes g:
A\otimes A' \to B\otimes B'$.  This operation of parallel combination is
called \emph{tensor}.

In the diagrammatic notation, composition is expressed by plugging the
outputs of one box into the inputs of another, and the tensor
is given by juxtaposition. 
\ctikzfig{cat_compose_tensor}
We require that both operations, composition and tensor, are
associative and obey the interchange law, 
\begin{equation}
(f \otimes g) \circ (h \otimes k) 
=
(f \circ h) \otimes (g \circ k)\;.\label{eq:interchangelaw}
\end{equation}
In the graphical notation, all of these equations become trivial: they
boil down the statement that the three diagrams below are unambiguous.
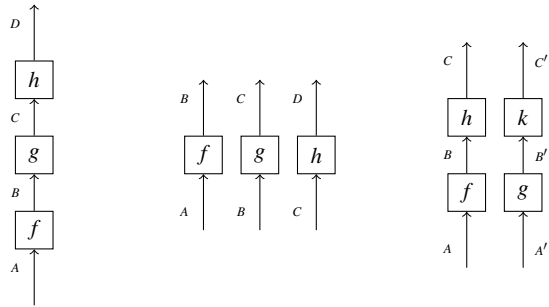
\ctikzfig{w2}
While it is
easy to translate these diagrams back into conventional notation, to
do so we must make a \emph{choice} of where to put the brackets, even
though the theory tells us this choice does not matter.  This
highlights a key advantage of working with diagrams, namely that the
objects which are \emph{equal} in the theory produce the \emph{same
  diagram}.

In addition to the two operations, composition and tensor, every
generalised compositional theory is equipped with certain primitive
processes.  The simplest process is the process which doesn't do
anything at all, simply returning unchanged the system given to it.
We assume that for every system $A$ such a null process, called the
\emph{identity} and written $\id{A}:A\to A$, exists.  The fact that it
does nothing is expressed by the equations
\[
\id{B} \circ f = f = f \circ \id{A}
\]
for all processes $f:A\to B$.  The identity process $\id{A}:A\to A$ is
drawn as a wire without any box on it, while the identity for $A
\otimes A'$ is simply the tensor product $\id{A} \otimes \id{A'}$,
i.e. two wires.
\[
\id{A} =   \idfig{A} \qquad \id{A\otimes A'} = %
\beginpgfgraphicnamed{tensor_wires}
\InputIfFileExists{tensor_wires.tikz}{}{\input{./figures/tensor_wires.tikz}}
\endpgfgraphicnamed
\]
Once again we see an equation absorbed into the notation:  since the
identity has no effect on a process, the \emph{length} of the wires attached
to a box makes no difference.
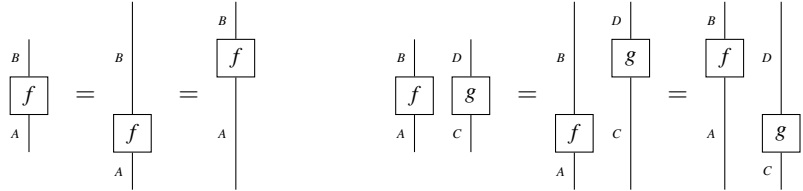
\ctikzfig{equiv_planar_diagrams}
In addition, for every pair of systems $A$ and $B$ there is a process
$\sigma_{A,B} : A \otimes B \to B \otimes A$ which exchanges the two
systems.  The class of theories we consider here are \emph{symmetric}:
swapping two systems twice has no effect, hence the equation
\[
\sigma_{B,A} \circ \sigma_{A,B} = \id{A \otimes B}
\]
holds for all systems $A$ and $B$
Graphically, the swap is just the crossing of two wires:
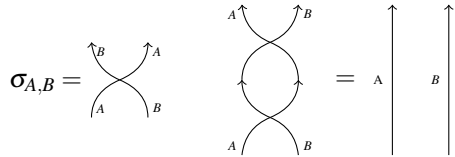
\[
\sigma_{A,B} = %
\beginpgfgraphicnamed{w3}
\InputIfFileExists{w3.tikz}{}{\input{./figures/w3.tikz}}
\endpgfgraphicnamed
\]
In fact, the swap should satisfy some further \emph{coherence
  equations}, the details of which can be found in
\cite{MacLane:CatsWM:1971}.  However, we can again make the graphical
notation do the work by allowing wires to cross freely in the
diagrams, and saying that only the connectivity of the wires matters,
and not their configuration in the page.  For example, the following
diagrams are equal: 
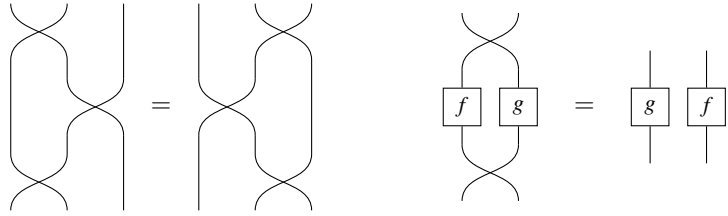
\ctikzfig{circuit_diagram_equiv} 
Note that we do not distinguish between wires crossing over and
crossing under.

A process may produce an output without having to consume
an input first, or vice versa.  Therefore we introduce a null system, or
empty system, which we denote $I$.  Hence a process that produces an $A$
from nothing would be written $p:I\to A$.  Like the identity process,
the null system obeys some equations:  
\[
A\otimes I = A = I \otimes A 
\qquad\text{ and } \qquad 
\id{I} \otimes f = f = f \otimes \id{I}\,,
\] 
for all systems $A$ and all processes $f$.  As suggested by the
preceding equations, $I$ is represented as empty space in the diagram,
and its identity process $\id{I}:I\to I$ is represented by the empty
diagram.  
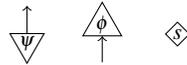
\ctikzfig{w4} 
A process of type $s:I\to I$ is called a \emph{scalar}; this name will
be justified later.  It is clear from the diagrammatic notation that
given scalars $s$ and $s'$ we have $s\circ s = s \otimes s' = s'
\otimes s = s' \circ s$; i.e. the scalars form a commutative
monoid\footnote{ This is true even for non-symmetric monoidal
  categories; see \cite{JS:1993:GeoTenCal1}.  }.

In the preceding text we have introduced various transformations of
diagrams which, we claim, do not change anything.  It is reasonable to
ask:  when are two diagrams considered to be equal?
We use a very intuitive notion here: \textbf{Two diagrams are
  considered equal when, keeping the inputs and outputs fixed, one may
  be transformed to the other by purely topological transformations.}
In other words, if starting from one diagram we---by crossing or
uncrossing wires, stretching wires, moving boxes along wires,
translating boxes in the plane (while maintaining their connections),
etc---arrive at the other, then they are equal.
\begin{figure}[ht]
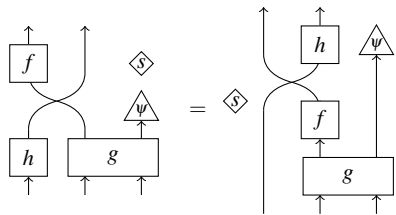

  \centering
  \ctikzfig{w5}
  \caption{Examples of topologically equivalent diagrams.}
  \label{fig:equality-of-diags}
\end{figure}
In particular, since scalars are not connected to the  inputs or
outputs of the diagram, they may be placed anywhere in the diagram
without altering its meaning.

\begin{example}\label{ex:symmetricgroups}
  The simplest non-trivial example is the theory with one primitive
  system, denoted $u$, and whose processes are generated by the
  identity and swap.  We call this theory \symgrp. Since there is only
  one basic system, every other system is just an $n$-fold tensor power of
  $u$, hence the systems of the theory can be identified with the
  natural numbers.  In this theory, a process $p:n\to n$ is nothing
  more than a sequence of swaps; i.e. a permutation on the $n$-element
  set.  Hence \symgrp is exactly the theory of the symmetric groups.
  \begin{figure}[ht]
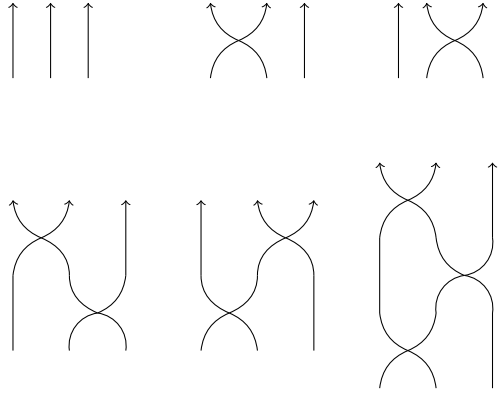

    \centering
    \ctikzfig{w6}
    \caption{Example: the symmetric group $S_3$ presented as diagrams.}
  \end{figure}
\end{example}

\begin{example}[Finite-dimensional Hilbert spaces]\label{ex:fdhilb}
  The theory called \fdhilb has as its systems all finite-dimensional
  complex Hilbert spaces.  The processes of this theory are all linear
  maps $f:A\to B$.  The sequential composition of processes is the
  usual composition of linear maps, and the tensor is the usual
  Kronecker product of vector spaces and maps.  The identity
  process is the identity map, the swap is the evident permutation map,
  and the null system is the base field, $\mathbb{C}$.  Since a linear
  map $\mathbb{C}\to\mathbb{C}$ is totally
  determined by its value at 1, we see that the scalars of \fhilb are
  nothing more than the complex numbers themselves.

  We write \fdhilbd to denote the subtheory \fdhilb restricted to
  Hilbert spaces of dimension $D^n$ and linear maps between them, for
  some fixed $D$.  For convenience, we refer to \fdhilbii as \qubit.
  Notice that the systems of \qubit are all tensor powers of
  $\mathbb{C}^2$, and its processes include all quantum circuits,
  state preparations,  and
  post-selected measurements, justifying the name.
\end{example}

\begin{remark}\label{rem:specify-tensor}
  Note that we must specify what the tensor product is to specify what
  the theory is.  For example, another equally valid theory is the
  collection of finite dimensional Hilbert spaces and linear maps, but
  with the direct sum as the tensor.  This is again a general
  compositional theory, although since it lacks certain other features
  we will require later, it will play no role in this presentation.
\end{remark}

\begin{example}[Sets and Relations]\label{ex:Rel}
  An example with very different flavour, but most of the same
  structure is \FRel.  The systems of \FRel are all finite sets
  (considered up to isomorphism\footnote{Since we identify sets of the
    same cardinality, we can equivalently say that the systems of
    \Frel are just the natural numbers.}), and the
  processes $r:X\to Y$ are relations between $X$ and $Y$, that is
  subsets of $X\times Y$.  The composition of relations is given by
  \[
  s\circ r = \{ (x,z) \;|\; \exists y \text{ s.t. } (x,y) \in r \text{ and
  } (y,z) \in s\}.
  \]
  The identity process is the  diagonal relation, 
  \[
  \id{X} = \{ (x,x) \;|\; x \in X\}.
  \]
  The tensor product in \Frel is the cartesian product $X\otimes
  Y = X \times Y$, which takes the form 
  \[
  r \otimes r' = \{((x,x'),(y,y')) \;|\; (x,y) \in r \text{ and } (x',y') \in
  r'\}
  \]
  on processes.  The null system is the singleton set $\{ * \}$, for
  which we have $\{*\} \times X \iso X$ for all sets $X$.  There are
  exactly two relations from $\{*\}$ to itself, namely the total
  relation and the empty relation.  Hence, the scalars of \frel are the
  Boolean monoid, i.e. $\mathbb{Z}_2$ with the usual multiplication.

  An important subtheory of \frel is \fset, obtained by restricting the
  to relations which are functions:  that is, relations $r:X\to Y$
  where each $x$ is related to exactly one $y$.
  Just as in the case of \fhilb, we can consider restrictions of \FRel
  to systems generated by a set of size $D$, which we call \freld.
  For example, \Frelii contains all the Boolean functions.  The
  intersection of \frelii and \fset consists of precisely the Boolean
  functions; this theory we denote \bool.
  Many other interesting theories are subtheories of \FRel; we'll meet some
  more later.
\end{example}

Since generalised compositional theories all share certain basic
structure, it is natural to consider maps between them.  Given two
such theories \catC and \catD, a map $F:\catC\to\catD$ 
consists of an assignment of each system $A$ in \catC to a system $FA$
in \catD, and an assignment of each process $f:A\to B$ in \catC to a process
$Ff:FA\to FB$ in \catD, obeying the following equations:
\begin{align*}
  F(A \otimes B) &= FA \otimes FB & FI &= I \\
  F(g\circ f) &= Fg \circ Ff & F(f \otimes g) &= Ff \otimes Fg \\
  F\id{A} &= \id{FA} & F\sigma_{A,B} &= \sigma_{FA,FB}
\end{align*}
In the mathematics literature, such a map is called a \emph{strict
  symmetric monoidal  functor}; again, see Mac
Lane \cite{MacLane:CatsWM:1971} for the details.  The important point
to note is that the mapping $F$ sends wires to wires.  Therefore, to
specify such a mapping it is enough to specify the image of the
boxes in a diagram, ensuring that composition and tensor are
respected.

\begin{example}\label{ex:rep-symgrp-in-fhilb}
  We can define a map $R_D: \symgrp \to \fdhilb$ by setting $R_D(u) =
  \mathbb{C}^D$ and then everything else is defined by the requirement
  that $R_D$ is a strict symmetric monoidal functor.  Thus we have a
  $D^n$ dimensional 
  representation of the symmetric group $S_n$ for every $D$.  

  In fact, this construction applies equally well to any generalised
  compositional theory \catC: all that is required is an assignment of the
  unique primitive system $u$ to some system of $\catC$.  Therefore
  every generalised compositional theory contains all the symmetric
  groups.
\end{example}

Given a mapping between theories it is easy to calculate the
image of a given diagram.  One must recursively partition the diagram
into tensors and compositions of smaller diagrams until each partition
contains exactly one element---that is, either a single wire, a
crossing of wires, or a box.  The interchange law (Equation~\ref{eq:interchangelaw}) guarantees that the
result does not depend on the partition chosen.

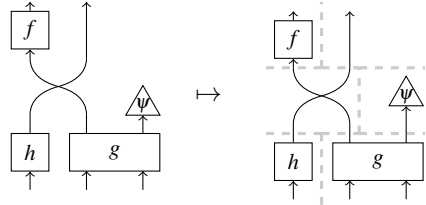
\ctikzfig{w7}

We may now state:
\begin{theorem}[Fundamental Theorem of Diagrams]
  \label{thm:fund-thm-diags-i}
  Given any two generalised compositional theories \catC and \catD,
  and a map $F:\catC \to \catD$, for any two diagrams $d$ and $d'$ in
  \catC, if $d = d'$ \emph{as diagrams} then $Fd = Fd'$ in \catD.
\end{theorem}
This theorem has many variations, and we refer the reader to
Selinger's survey article \cite{SelingerSurvey} for the full details.

\begin{remark}\label{rem:1}
In the diagrams to come, we will often use horizontal separation to
indicate separation in space and vertical separation to indicate
separation in time. For example,  
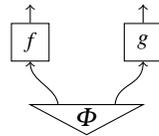
\[
\beginpgfgraphicnamed{spacetime}
\InputIfFileExists{spacetime.tikz}{}{\input{./figures/spacetime.tikz}}
\endpgfgraphicnamed 
\]
depicts the creation of two systems by the process $\Phi$, which then
become spatially separated over some time and are acted upon by
processes $f$ and $g$ respectively.  Since, as we already know,
topologically equivalent diagrams are equal, these separations have no
formal status and are purely illustrative.
\end{remark}


\subsection{Rewrites and Models}
\label{sec:mapp-rewr-models}

Since we wish to generalise over many concrete mathematical structures,
we are particularly interested in theories which can be specified
\emph{axiomatically}.  That is, to specify the theory we state (i) the
list of basic systems---typically we'll only have one basic system,
the rest being generated by the tensor product---and (ii) the basic
processes.  The processes of the theory are then \emph{all} the
diagrams which can be constructed from these processes and nothing
else.

\begin{example}[Boolean Circuits]
  \label{ex:boolean-circuits}
  A simple example of a compositional theory is \boolcircs, the theory
  of boolean circuits.  This theory has only one basic system, the bit
  $b$, and the basic processes are the logic gates:
  \[
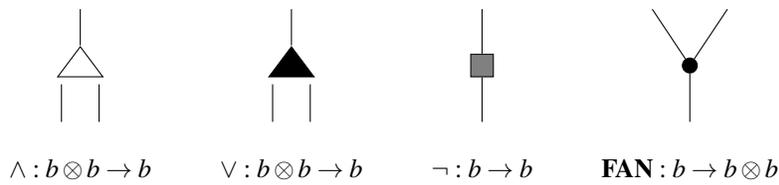
  \begin{array}{ccccccc}
\beginpgfgraphicnamed{bool-and}
\InputIfFileExists{bool-and.tikz}{}{\input{./figures/bool-and.tikz}}
\endpgfgraphicnamed & &
\beginpgfgraphicnamed{bool-or}
\InputIfFileExists{bool-or.tikz}{}{\input{./figures/bool-or.tikz}}
\endpgfgraphicnamed & &
\beginpgfgraphicnamed{bool-not}
\begin{tikzpicture}
	\begin{pgfonlayer}{nodelayer}
		\node [style=notgate] (0) at (0, -0) {};
		\node [style=none] (1) at (0, 0.75) {};
		\node [style=none] (2) at (0, -0.75) {};
	\end{pgfonlayer}
	\begin{pgfonlayer}{edgelayer}
		\draw [style=none] (0) to (1.center);
		\draw (2.center) to (0);
	\end{pgfonlayer}
\end{tikzpicture}}
\endpgfgraphicnamed & &
\beginpgfgraphicnamed{bool-fan}
\InputIfFileExists{bool-fan.tikz}{}{\input{./figures/bool-fan.tikz}}
\endpgfgraphicnamed 
  \\ \\
  \mathrm{\wedge}:b \otimes b \to b &\qquad  &
  \mathrm{\vee}:b \otimes b \to b & \qquad &
  \mathrm{\neg}: b \to b & \qquad &
  \textbf{FAN} : b \to b \otimes b
  \end{array}
  \]
  A process in this theory is a circuit for computing some boolean
  function, built up from these basic gates.

  \begin{figure}[ht]
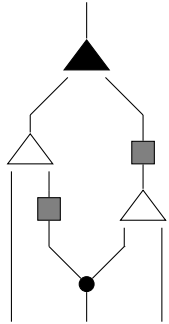

    \centering
    \ctikzfig{w9}
    \caption{A Boolean circuit to compute $(x \wedge \neg y) \vee \neg (y \wedge z)$.}
\label{fig:boolcirc-example}
\end{figure}

  It is tempting to assume that \boolcircs is related to the theory of
  Boolean functions, and we can make this precise by specifying a
  mapping $B:\boolcircs\to\bool$.  We assign $B(b) = \{0,1\}$ and
  define $B$ on the basic processes as follows:
  \begin{align*}
    B( \wedge ) &= a:\left\{\begin{array}{rcl} 
      00 & \mapsto & 0\\
      01 & \mapsto & 0\\
      10 & \mapsto & 0\\
      11 & \mapsto & 1
    \end{array}\right. &
    B( \vee ) &= o:\left\{\begin{array}{rcl} 
      00 & \mapsto & 0\\
      01 & \mapsto & 1\\
      10 & \mapsto & 1\\
      11 & \mapsto & 1
    \end{array}\right. \\
    B( \neg ) &= n:\left\{\begin{array}{rcl} 
      0 & \mapsto & 1\\
      1 & \mapsto & 0
    \end{array}\right. &
    B( \textbf{FAN} ) &= \delta:\left\{\begin{array}{rcl} 
      0 & \mapsto & 00\\
      1 & \mapsto & 11
    \end{array}\right.
  \end{align*}
\end{example}
The mapping $B$ assigns to each diagram the boolean function normally
associated with it.  However this is not the only possibility.
Consider the following mapping, $P:\boolcircs\to\bool$.  Once again
$P(b) = \{0,1\}$, but now we have the following assignment of
processes:
  \begin{align*}
    P( \wedge ) &= a:\left\{\begin{array}{rcl} 
      00 & \mapsto & 0\\
      01 & \mapsto & 0\\
      10 & \mapsto & 0\\
      11 & \mapsto & 1
    \end{array}\right. &
    P( \vee ) &= p:\left\{\begin{array}{rcl} 
      00 & \mapsto & 0\\
      01 & \mapsto & 1\\
      10 & \mapsto & 1\\
      11 & \mapsto & 0
    \end{array}\right. \\
    P( \neg ) &= i:\left\{\begin{array}{rcl} 
      0 & \mapsto & 0\\
      1 & \mapsto & 1
    \end{array}\right. &
    P( \textbf{FAN} ) &= \delta:\left\{\begin{array}{rcl} 
      0 & \mapsto & 00\\
      1 & \mapsto & 11
    \end{array}\right.
  \end{align*}
 The mapping $P$ assigns to each $d:b^n \to b$ in \boolcircs an $n$-variable
 polynomial over the ring $\mathbb{Z}_2$.  (More generally a circuit
 with multiple outputs produces a list of polynomials, one for each
 output.)

In fact, as the example of $P$ suggests, the diagrams of \boolcircs
admit an interpretation in any setting with two binary operations and
one unary operation.  This is not entirely satisfactory.  In order to
capture more than the bare syntax of any given theory we need to
impose some additional equations on the class of diagrams.  We do this
via \emph{rewrite rules}.

A rewrite rule consists of a pair of diagrams of the same type, for
example $d:A\to B$ and $d':A\to B$.  If this rule is called $r$ then
we write $r:d \rw d'$, or diagrammatically
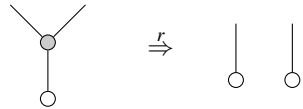
\ctikzfig{w10}
Whenever $d$ occurs as a subdiagram of a larger diagram
$e$ then we can replace $d$ with $d'$ in $e$, written $e[d] \RW{r}
e[d']$, or in diagrams:
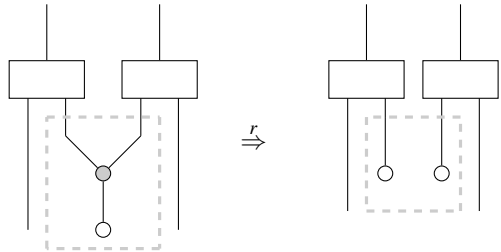
\ctikzfig{w11}
Rewrite rules allow us to define a notion of equality in addition to the
basic equality of diagrams.  Given a collection of rewrite rules
$\mathcal{R}$ we write $d\RW{\mathcal{R}} d'$ if there is some rewrite
sequence 
in $\mathcal{R}$ taking $d$ to $d'$.  Evidently $\RW{\mathcal{R}}$ is
a transitive relation;  let $\EQ{\mathcal{R}}$ be its symmetric,
reflexive closure.  Then we say that two processes are equal according
to $\mathcal{R}$ if their corresponding diagrams satisfy $d
\EQ{\mathcal{R}} d'$.  Typically
we'll exhibit this equivalence as a sequence of rewrites.

 \begin{example}[Boolean circuits]\label{ex:boolcircs-rw}
   Consider the following two rewrite rules for \boolcirc, expressing
   respectively the distributivity of AND over OR, and (one half of)
   De Morgan's law.
   \[
\beginpgfgraphicnamed{w12}
\InputIfFileExists{w12.tikz}{}{\input{./figures/w12.tikz}}
\endpgfgraphicnamed \qquad\qquad\qquad    %
\beginpgfgraphicnamed{w13}
\InputIfFileExists{w13.tikz}{}{\input{./figures/w13.tikz}}
\endpgfgraphicnamed
   \]
   Now we can show that a certain Boolean circuit can be transformed
   into its disjunctive normal form:
   \ctikzfig{w14}
\end{example}

Given a theory $\catC$, a set of rewrite rules $\mathcal{R}$, and
a mapping $F:\catC\to\catD$, we
can ask the following question: if $ d \EQ{\mathcal{R}} d'$ in \catC,
is it the case that $F d = F d'$ in \catD?

This property is called \emph{soundness}.  A sound  mapping $F:\catC \to
\catD$ is called an \emph{interpretation} of \catC
in \catD, and the image of \catC in \catD is called a \emph{model}.
In the example above, the mapping $B$ is sound, hence it provides an
interpretation of \boolcirc (and $\mathcal{R}$) in \bool; on the other
hand $P$ does not, due to the failure of De Morgan's law.  Generally
speaking we will always work with a given set of rewrite rules
and a given interpretation map, so we will usually say ``\emph{the}
\catD interpretation of \catC'', although in principle there could be
many.

\begin{remark}\label{rem:completeness}
  The converse property to soundness, $F d = F d'$ implies $d = d'$,
  is called \emph{completeness}.  An interpretation which is both
  sound and complete provides an isomorphism between the formally
  presented theory and its model.  While checking soundness is
  straightforward, showing completeness is often much more
  difficult\footnote{ 
    To show completeness for a rewrite theory it is typically
    necessary, but rarely sufficient, to check that the rewrite rules
    are \emph{confluent}; that is, whenever two rewrites
    simultaneously apply to a given diagram, then the choice between
    then (eventually) does not matter.  Since this property must hold
    for every diagram and every pair of rewrites, even a simple
    rewrite system can produce an extremely large number of cases,
    necessitating a computer-assisted proof.  For example see the work
    of Lafont on Boolean circuits \cite{Lafont:2003qy}.  
  }.  On the
  other hand, not having completeness means there are multiple models
  of a given theory, and the study of the differences between such
  models is often informative.
\end{remark}

Before moving on, we'll introduce an important example, and its
standard model.

\begin{example}[Quantum Circuits]\label{ex:qcircuits}
  Similar to the example of Boolean circuits, we can also view
  (post-selected) quantum circuits as generalised compositional
  theory, called \qcirc.  Again we have a single basic system, the
  qubit $Q$, and the basic processes are a collection of unitary
  gates, state preparations, and projections from which we construct
  the other quantum circuits.
  
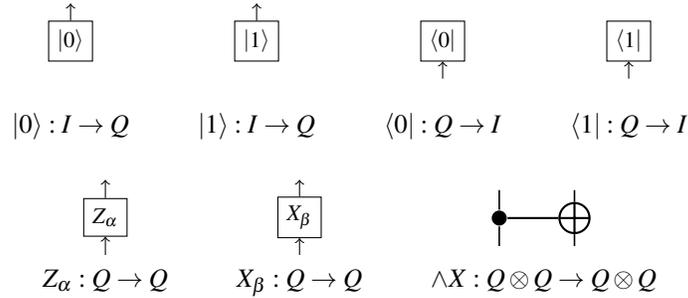
\[
  \begin{array}{ccccccc}
    \boxpoint{\ket0} && 
    \boxpoint{\ket1} &&
    \boxcopoint{\bra0} && 
    \boxcopoint{\bra1}\\
    \\
    \ket0 : I \to Q & \qquad &
    \ket1 : I \to Q & \qquad &
    \bra0 : Q \to I & \qquad &
    \bra1 : Q \to I \\
  \end{array}
  \]
  \[
  \begin{array}{ccccc}
    \boxmap{Z_\alpha} &  \qquad &
    \boxmap{X_\beta} & \qquad &
    \circCX \\
    Z_\alpha : Q \to Q & \qquad &
    X_\beta : Q\to Q & \qquad &
    \CX : Q\otimes Q \to Q \otimes Q
  \end{array}
  \]
  From these basic elements we can write down any quantum circuit.  We
  now define the standard interpretation of \qcirc into \qubit.
  \[
  \denote{Q} = \mathbb{C}^2
  \]
  \begin{align*}
    \denote{\,\boxpoint{\ket0}\,} &= \ket{0} &
    \denote{\,\boxpoint{\ket1}\,} &= \ket{1} &
    \denote{\,\boxcopoint{\bra0}\,} &= \bra{0} &
    \denote{\,\boxcopoint{\bra1}\,} &= \bra{1} \\
  \end{align*}
  \begin{align*}
    \denote{\;\boxmap{Z_{\alpha}}\;} &=
    \begin{pmatrix*}
      1 & 0 \\ 0 & e^{i\alpha}
    \end{pmatrix*} &
    \denote{\;\boxmap{X_{\beta}}\;} &=
    \begin{pmatrix*}[r]
      \cos\frac{\beta}{2} & -i \sin \frac{\beta}{2} \\ 
      -i \sin \frac{\beta}{2} & \cos\frac{\beta}{2}  
    \end{pmatrix*}
  \end{align*}
  \[
  \denote{\;\circCX\;} \;=\;
  \begin{pmatrix*}
    1 & 0 &0&0 \\ 0 & 1 & 0 & 0 \\ 0 & 0 & 0 & 1 \\ 0&0&1&0
  \end{pmatrix*}
  \]
  Thanks to the well-known universality result
  \cite{Adriano-Barenco:1995qy} this interpretation demonstrates that
  \qcirc can represent all unitary maps between qubits.  In fact,
  since we have the projections $\bra{0},\bra{1}$, all linear maps can
  be represented.  Note, however, that although all quantum circuits
  can be represented, without a set of rewrite rules \qcirc cannot
  express any non-trivial equalities between them.  We could propose
  various sound equations here, but there is no known collection of
  rewrite rules which makes \qcirc complete with respect to this
  interpretation into \qubit.  If such a set of rewrites did exist, it
  would constitute provide a presentation of the unitary group by
  generators and relations.
\end{example}


\subsection{The Dagger}
\label{sec:dagger}

Now we introduce the \emph{dagger}.  This is simply an operation on
the processes of a theory, sending every process $f:A\to B$ to another
process $f^\dag:B\to A$. We call $f^\dag$ the \emph{adjoint} of $f$.
In the graphical calculus, we represent the dagger by a flip in the
horizontal axis:
\[ 
\left(\map{f}\right)^\dagger = \mapdag{f} 
\] 
Note that we have made the box asymmetric to make this flipping
evident.  For more general diagrams, the dagger flips a diagram upside
down, preserving all the internal structure.  Taking this claim at
face value, we can derive the key properties of the dagger:
\begin{align*}
  (f^\dag)^\dag &= f 
  & \left(\mapdag{f}\right)^\dag  \;&=\; \map{f} 
  \\ \\
  (g\circ f)^\dag &= f^\dag \circ g^\dag 
  & \left( %
\beginpgfgraphicnamed{dagger-comp-lhs}
\InputIfFileExists{dagger-comp-lhs.tikz}{}{\input{./figures/dagger-comp-lhs.tikz}}
\endpgfgraphicnamed \right)^\dag \;&=\;
\beginpgfgraphicnamed{dagger-comp-rhs}
\InputIfFileExists{dagger-comp-rhs.tikz}{}{\input{./figures/dagger-comp-rhs.tikz}}
\endpgfgraphicnamed
  \\ \\
  (f \otimes g)^\dag &= f^\dag \otimes g^\dag 
  & \left( \map{f} \; \map{g} \right)^\dag \;&=\; \mapdag{f} \;
  \mapdag{g}
  \\\\
  \id{A}^\dag &= \id{A} 
  & \left( \idfig{A}\; \right)^\dag \;&=\; \idfig{A} 
\\\\
%
%
%
\sigma_{A,B}^\dag &= \sigma_{B,A} 
  & \left( \; %
\beginpgfgraphicnamed{swap}
\InputIfFileExists{swap.tikz}{}{\input{./figures/swap.tikz}}
\endpgfgraphicnamed\;\right)^\dag \;&=\; %
\beginpgfgraphicnamed{swap}
\InputIfFileExists{swap.tikz}{}{\input{./figures/swap.tikz}}
\endpgfgraphicnamed
\end{align*}
The dagger allows two important concepts to be defined.
\begin{definition}\label{def:unitary}
  A process $f:A\to B$ is called \emph{unitary} if $f\circ f^\dag =
  \id{B}$ and $f^\dag \circ f = \id{A}$.  A process is called
  \emph{self-adjoint} when $f^\dag = f$.
\end{definition}

\begin{example}[Finite-dimensional Hilbert spaces] \label{ex:unitary-fdhilb}
  The theory \fhilb admits a dagger:  it is the usual adjoint of a
  linear map.  In this theory, the abstract definitions of unitarity
  and self-adjointness coincide with the usual one.
\end{example}

\begin{example}\label{ex:unitarity-frel}
  In the theory \frel, the dagger of a relation $r:X\to Y$ is defined
  by the converse relation, i.e.
  \[
  r^\dag = \{ (y,x) \;|\; (x,y) \in r\}
  \]
  Here, unitary processes are exactly those relations which encode
  permutations.  A relation is self-adjoint whenever it is symmetric.
  Hence the self-adjoint unitaries in \frel are exactly the
  permutations of order 2.
\end{example}

We extend the definition of mapping to demand that it also preserves
the dagger.  That is, given two theories with dagger, we require that
a map $F:\catC \to \catD$ satisfies
\[
F(f^\dag) = (Ff)^\dag
\]

\begin{example}[Quantum Circuits]\label{ex:unitary-qcirc}
  We define a dagger on \qcirc as follows:
  \begin{align*}
    \left( \boxpoint{\ket0} \right)^\dag &= \boxcopoint{\bra0} &
    \left( \boxpoint{\ket1} \right)^\dag &= \boxcopoint{\bra1}\vspace{2mm} \\
    \left( \boxcopoint{\bra0}\right)^\dag &=  \boxpoint{\ket0}&
    \left( \boxcopoint{\bra1}\right)^\dag &=  \boxpoint{\ket1}\vspace{2mm}\\
    \left( \boxmap{X_{\alpha}} \right)^\dag &= \boxmap{X_{-\alpha}} &
    \left( \boxmap{Z_{\beta}}\right)^\dag &=  \boxmap{Z_{-\beta}}
  \end{align*}
  \[
  \left( \circCX\right)^\dag =  \circCX
  \]
  It's now easy to check that the interpretation map introduced
  earlier, $\denote{\cdot}:\qcirc \to \qubit$ preserves the dagger as
  required.
\end{example}

\begin{remark}\label{rem:no-dagger-for-boolcircs}
  The theory of Boolean circuits, \boolcirc, does not admit a dagger.
  However, we could formally add new basic processes corresponding to
  the adjoints of the basic processes of \boolcirc and thus define a
  new theory, $\boolcirc^\dag$.  Since the converse of a function is
  not in general a function, the interpretation $B:\boolcirc \to
  \bool$ no longer makes sense.  Instead we must interpret
  $\boolcirc^\dag$ over \frelii, that is as Boolean relations rather
  than functions.  In this case $B$ again defines a valid
  interpretation $\boolcirc^\dag \to \frelii$.  The resulting theory
  is a model of non-deterministic computation.
\end{remark}

In any theory, a process of type $p:I\to A$ is called a \emph{point},
or sometimes a \emph{state}, of $A$.  Dually, a process of type $e:A
\to I$ is called a \emph{co-point}, or sometimes an \emph{effect} on
$A$.  For example, in \fdhilb the points $\psi:I\to A$ are in
one-to-one correspondence with the vectors of $A$, while in \frel a
point $s:I\to X$ is precisely a subset of $X$.

In a theory with a dagger the set of points is isomorphic to the set
of copoints (or in other language, for every state there is a
corresponding effect and vice versa).  This allows us to define
another important concept.

\begin{definition}\label{def:inner-prod}
  Given two points $\psi,\phi :I \to A$ we define their \emph{inner
    product} as $\phi^\dag \circ \psi$.  Dually, the outer product is
  defined as $\phi \circ \psi^\dag$.
\end{definition}

As one may expect, the inner product is always a scalar. The
diagrammatic language automatically allows the same tricks---and
more--- as Dirac notation does in Hilbert spaces.  Indeed one can view
the diagrammatic language as a 2-dimensional generalisation of Dirac
notation.

\begin{example}[Finite-dimensional Hilbert spaces]
  In \fdhilb the inner product defined by the dagger, is exactly
  the usual inner product $\innp{\phi}{\psi}$.
\end{example}
\begin{example}[Sets and Relations]
  In \frel the inner product $r^\dag \circ s$ is 0 if the $r$
  and $s$ are disjoint as subsets, and 1 otherwise.
\end{example}

\section{Pure state quantum mechanics}

\subsection{The elements of an operational theory}
\label{sec:elem-an-oper}

It is remarkable that the the basic language of quantum
mechanics---states, effects, unitarity, self-adjointness, inner
products, tensor products---can all be defined in the abstract setting
of generalised compositional theories.  We now have enough material to
describe a formal operational framework for pure state quantum
mechanics in purely diagrammatic terms.

\begin{itemize}
\item A \emph{preparation} is any process which produces a
  state;  that is to say it is process of type $p:I\to A$.
  \[
\beginpgfgraphicnamed{point}
\begin{tikzpicture}[dotpic]
	\begin{pgfonlayer}{nodelayer}
		\node [style=point] (0) at (0, -0.5) {$p$};
		\node [style=none] (1) at (0, 0.5) {};
	\end{pgfonlayer}
	\begin{pgfonlayer}{edgelayer}
		\draw [swap] (0) to node[wire label]{$A$} (1.center);
	\end{pgfonlayer}
\end{tikzpicture}}
\endpgfgraphicnamed
  \]
  Preparations are not restricted to producing single systems; a
  preparation process of type $I\to A_1 \otimes \cdots A_n$ is called
  \emph{multipartite}.  Of course, multipartite preparations need not
  be separable.
  \[
\beginpgfgraphicnamed{pointpmulti1}
\InputIfFileExists{pointpmulti1.tikz}{}{\input{./figures/pointpmulti1.tikz}}
\endpgfgraphicnamed\qquad\qquad%
\beginpgfgraphicnamed{pointpmulti2}
\InputIfFileExists{pointpmulti2.tikz}{}{\input{./figures/pointpmulti2.tikz}}
\endpgfgraphicnamed
  \]
  When interpreted in \fdhilb each preparation process yields a ray in
  some Hilbert space, which, ignoring global phase, we may identify
  with a specific quantum state.  It may happen, depending on the
  equations of the formal theory, that different preparation
  processes produce the same state.
\item A \emph{transformation} is any process which
  acts on states and produces new states, and which is unitary:
  \[
  \boxmap{U}
  \]
  Once again, transformations may act on one or many systems at the
  same time.
  \ctikzfig{big-unitary}
\item Measurements are processes which accept quantum
  inputs and produce classical information about the state which was
  input.  Since, for now, our theory only has pure states, we will
  work with \emph{non-degenerate} \emph{post-selected}
  measurements\footnote{
    In other words, rank 1 projectors.
};
  i.e. we know that a definite outcome has occurred, and that outcome
  corresponds to a definite quantum state.  Therefore, measurements
  are one-dimensional effects, represented as co-points:
  \[
\beginpgfgraphicnamed{copoint}
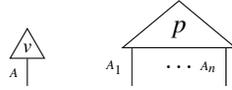
\begin{tikzpicture}[dotpic]
	\begin{pgfonlayer}{nodelayer}
		\node [style=copoint] (0) at (0, 0.5) {$v$};
		\node [style=none] (1) at (0, -0.5) {};
	\end{pgfonlayer}
	\begin{pgfonlayer}{edgelayer}
		\draw [swap] (0) to node[wire label]{$A$} (1.center);
	\end{pgfonlayer}
\end{tikzpicture}}
\endpgfgraphicnamed\qquad  %
\beginpgfgraphicnamed{copointpmulti2}
\InputIfFileExists{copointpmulti2.tikz}{}{\input{./figures/copointpmulti2.tikz}}
\endpgfgraphicnamed
  \]
  The classical information is implicit in the choice of
  copoint, and hence not represented.  Since copoints do not have
  quantum outputs, these processes correspond to \emph{demolition
    measurements}, where the original system is consumed by the
  measurement process.  However, by combining an effect with the
  corresponding state preparation we can also represent non-demolition
  measurements:
  \[
\beginpgfgraphicnamed{projectordecomp}
\InputIfFileExists{projectordecomp.tikz}{}{\input{./figures/projectordecomp.tikz}}
\endpgfgraphicnamed \qquad   %
\beginpgfgraphicnamed{projectordecompmulti}
\InputIfFileExists{projectordecompmulti.tikz}{}{\input{./figures/projectordecompmulti.tikz}}
\endpgfgraphicnamed
  \]
  To properly represent the non-determinism of quantum measurements we
  need to consider mixed states; this is dealt with in Section
  \ref{sec:mixed-stat-meas}.  More general measurements can be
  represented within the theory, however they will not be described
  here.
\end{itemize}

This basic recipe---preparations, transformations, and
measurements---allows any experimental setup to be described in terms
of the processes which realise it.  More precisely, since we use
post-selected measurements, the diagram really represents a \emph{run}
of the experiment where a certain outcome occurred.  We call an
experiment \emph{closed} when it has no external inputs or outputs.
Any closed experiment is necessarily described by a process of type
$x:I \to I$; that is, a scalar.  This scalar is the abstract
counterpart to the probability amplitude for performing the process
and observing the specified result.  Indeed, when such a diagram is
interpreted in \fdhilb, the result is exactly the probability
amplitude.

\begin{example}\label{ex:quantum-circs-bis}
  The theory \qcircs has the structure described above, and we
  can use it to define a simple experiment.  For example, the diagram
  below corresponds to preparing a qubit in the $\ket{0}$ state,
  applying a unitary gate to it, and, upon measuring in the
  computational basis, finding that the qubit is in the state
  $\ket{1}$.
  \ctikzfig{experiment-i}
  Using the interpretation map $\denote{\cdot}:\qcirc\to\qubit$ we can
  calculate the amplitude for this experimental result.
  \[
\beginpgfgraphicnamed{experiment-ii}
\InputIfFileExists{experiment-ii.tikz}{}{\input{./figures/experiment-ii.tikz}}
\endpgfgraphicnamed \quad\mapsto\quad \bra{1} \circ \frac{1}{\sqrt{2}}
  \begin{pmatrix*}[r]
      1 & -i  \\  -i  & 1  
    \end{pmatrix*} 
  \circ \ket{0} \quad= \quad \frac{-i}{\sqrt{2}}
  \]
\end{example}

To summarise the elements of the framework, a formal generalised
compositional theory consists of:
\begin{itemize}
\item A collection of basic systems and processes, corresponding to
  the available ``lab equipment''.
\item The collection of all diagrams constructed from the basic
  processes, corresponding to every possible experiment that could be
  built from the given equipment.  We consider diagrams modulo
  topological equivalence: equivalent diagrams correspond to the same
  experiment.
\item A (possibly empty) collection of axioms, presented as rewrite rules
  over diagrams, which specify behavioural equivalence of processes.
  These rules tell us when a piece of the experimental setup can safely
  be replaced by another without changing the result of the
  experiment.
\item Finally, given the above, we'll usually consider (sound)
  \emph{interpretation maps} of the formal theory into some concrete
  mathematical structure, such as Hilbert spaces.
\end{itemize}
So far we have been operating at an extremely high level of
generality.  To focus our attention on quantum systems we will now
gradually introduce more structure to our theories.  We 
identify certain structural features of the Hilbert space
presentation of quantum mechanics, and provide an abstract
realisation of those features in terms of basic processes and
equations, whose behaviour reproduces various quantum phenomena in
the abstract setting of generalised compositional theories.

The rest of this section will layout which
basic processes and equations we will need to realise.
As we do so, we'll say goodbye to
some of the models introduced earlier, but the two most important
ones, \fhilb and \frel, will still be applicable.

\subsection{ Duals}
\label{sec:duals}

The next piece of structure that will be required is the
existence of \emph{duals}\footnote{For the experts in category theory,
  this additional structure can be summed up by saying we operate in a
  dagger-compact category, rather than just a symmetric monoidal
  category.}.

\begin{definition}\label{def:duals}
  A system $A$ has a \emph{dual} if there exists a system $A^*$ and
  processes 
  \[
  e_A:I\to A^*\otimes A \quad \text{ and } \quad d_A:A \otimes A^* \to I
  \]
  such that  we have the following equations:
  \[
  (d_A \otimes \id{A}) \circ (\id{A} \otimes e_A) = \id{A}
  \qquad
  (\id{A^*} \otimes d_A) \circ (e_A \otimes \id{A^*}) = \id{A^*}
  \]
\end{definition}
Since this definition is rather hard to parse we will immediately move
to its diagrammatic form.  We indicate the dual system $A^*$ by a wire
labelled by $A$ but directed in the opposite direction.  The maps
$d_A$ and $e_A$ are represented by wires with half turns, henceforth
``caps'' and ``cups''.  The equations above then take the form of
``straightening wires'':
\ctikzfig{cat_cap_cup}
In general a system might have more than one dual, but they are all
guaranteed to be isomorphic.  We'll assume that every system has a
given dual, and in particular $(A \otimes B)^* = B^* \otimes A^*$, in
which case $d_{A\otimes B}$ and $e_{A \otimes B}$ take the form of
nested caps and cups.  Furthermore, we'll assume that the double dual
$A^{**} = A$.  These simplifications automatically hold in any theory
presented diagrammatically; taking them as the general case saves a
lot of bureaucracy.

\begin{example}[Finite dimensional Hilbert spaces]
  Let $A$ be a Hilbert space of dimension $d$, then $A^*$ is the usual
  dual space; that is, the space of linear functionals from $A$ to the
  complex numbers.  Supposing that $\{\ket{a_i}\}$ is a basis for the
  space $A$, then the cup and cap are given by the linear maps
  \[
  e_A: 1\mapsto \sum_i \bra{a_i}\otimes \ket{a_i}\;,
  \qquad
  d_A :  \sum_i \ket{a_i} \otimes \bra{a_i} \mapsto 1\;.
  \]
  However since we are in a finite-dimensional setting we could also
  choose $A^* =A$;  specialising to the case of qubits, we can now
  view the cup and cap as the preparation and projection onto a Bell
  state:
  \[
  e_Q = \ket{00} + \ket{11} 
  \qquad
  d_Q = \bra{00} + \bra{11} 
  \]
  Recall the quantum teleportations protocol: Alice has some unknown
  state that she wishes to send to Bob, but they do not share a
  quantum channel.  However they have a classical channel, and have
  previously shared a Bell pair.  In order to send her qubit to Bob,
  Alice measures her two qubits in the Bell basis, and transmits the
  result to Bob.  Now Bob simply applies some unitary map (depending
  on Alice's outcome) to his half of the Bell pair to recover the
  qubit that Alice wanted to send.  Since, for the moment, we are
  operating in a post-selected setting, we'll assume that Alice
  observes the outcome corresponding to the state $\Phi_+ = (\ket{00}
  + \ket{11})/\sqrt2)$ at her measurement.  In this case Bob need do
  nothing to his qubit.  The whole set up is shown below:
  \ctikzfig{teleport-i} 
  Knowing that the projection onto $\Phi_+$ is
  just the effect $d_Q$, we can rewrite the protocol as shown and
  demonstrate the protocol purely diagrammatically:
  \ctikzfig{teleport-ii} 
\end{example}

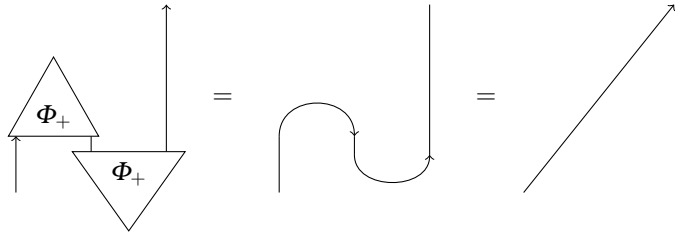
\begin{example}[Sets and Relations]
  In \frel, the  dual of a set $X$ is just the same set $X$ again.
  The cup is given by the ``name'' of the identity:
  \[
  e_X = \{ (*,(x,x)) \;|\; x \in X\}
  \]
  while the cap, $d_X$ is just the converse of $e_X$.
\end{example}

Using caps and cups, we can turn any process $f : A \rightarrow B$
into a process on the dual objects going in the opposite direction:
$f^* : B^* \rightarrow A^*$. 
\ctikzfig{bend_wires}
This is sometimes called the transpose of $f$, but this terminology
can be misleading. In $\catFHilb$, $f^*$ is the map that takes a
linear form $\bra\xi \in B^*$ to $\bra\xi\!f \in A^*$. We refer to
this map simply as the \textit{upper-star} of $f$.  Clearly, we have
$f^{**} = f$.  

It is also required that the dagger and the duals interact nicely.
More precisely we have the equations:
\[
\left(\, %
\beginpgfgraphicnamed{cup}
\begin{tikzpicture}[dotpic]
	\begin{pgfonlayer}{nodelayer}
		\node [style=none] (0) at (-0.75, 0) {};
		\node [style=none] (1) at (0.5, 0) {};
	\end{pgfonlayer}
	\begin{pgfonlayer}{edgelayer}
		\draw [style=diredge, in=-90, out=-90, looseness=1.75] (0.center) to (1.center);
	\end{pgfonlayer}
\end{tikzpicture}}
\endpgfgraphicnamed \,\right)^\dag \;=\;\; %
\beginpgfgraphicnamed{cap-twist}
\InputIfFileExists{cap-twist.tikz}{}{\input{./figures/cap-twist.tikz}}
\endpgfgraphicnamed
\qquad
\left(\, %
\beginpgfgraphicnamed{cap-FU}
\begin{tikzpicture}[dotpic]
	\begin{pgfonlayer}{nodelayer}
		\node [style=none] (0) at (-0.75, 0) {};
		\node [style=none] (1) at (0.5, 0) {};
	\end{pgfonlayer}
	\begin{pgfonlayer}{edgelayer}
		\draw [style=diredge, in=90, out=90, looseness=1.75] (0.center) to (1.center);
	\end{pgfonlayer}
\end{tikzpicture}}
\endpgfgraphicnamed \,\right)^\dag \;=\;\; %
\beginpgfgraphicnamed{cup-twist}
\InputIfFileExists{cup-twist.tikz}{}{\input{./figures/cup-twist.tikz}}
\endpgfgraphicnamed
\]
In any theory with both a dagger and duals, we can define a third
operation, the \textit{lower-star} of $f$ as $f_* := (f^\dagger)^* =
(f^*)^\dagger$.  Again this is involutive, i.e. $f_{**} = f$.  We'll
return to the uses of the upper and lower stars in Section
\ref{sec:mixed-stat-meas}.

Finally, the cup and cap can be used to define a \emph{trace} in purely
diagrammatic terms:
\[
\Tr (\map{f} ) = %
\beginpgfgraphicnamed{trace-i}
\InputIfFileExists{trace-i.tikz}{}{\input{./figures/trace-i.tikz}}
\endpgfgraphicnamed
\]
Checking the Hilbert space interpretation, it is easy to see that this
coincides with the usual definition.  
\[
\Tr( \map{f} ) = (\sum_i \bra{ii}) \circ (\id{A} \otimes f) \circ (\sum_j
\ket{jj})
= \sum_i \bra{i} f \ket{i} = \sum_i f_{ii}
\]
In the diagrammatic form it is
trivial to prove that trace is invariant under cyclic permutation:
\[
\Tr\left( \; %
\beginpgfgraphicnamed{dagger-comp-lhs}
\InputIfFileExists{dagger-comp-lhs.tikz}{}{\input{./figures/dagger-comp-lhs.tikz}}
\endpgfgraphicnamed \;\right)
= %
\beginpgfgraphicnamed{trace-ii}
\InputIfFileExists{trace-ii.tikz}{}{\input{./figures/trace-ii.tikz}}
\endpgfgraphicnamed 
= %
\beginpgfgraphicnamed{trace-iii}
\InputIfFileExists{trace-iii.tikz}{}{\input{./figures/trace-iii.tikz}}
\endpgfgraphicnamed
= %
\beginpgfgraphicnamed{trace-iv}
\InputIfFileExists{trace-iv.tikz}{}{\input{./figures/trace-iv.tikz}}
\endpgfgraphicnamed
= \Tr \left(\;%
\beginpgfgraphicnamed{trace-v}
\InputIfFileExists{trace-v.tikz}{}{\input{./figures/trace-v.tikz}}
\endpgfgraphicnamed\;\right)
\]
The partial trace can be  defined analogously.
\[
\Tr^A_B \left(\; %
\beginpgfgraphicnamed{trace-partial-i}
\InputIfFileExists{trace-partial-i.tikz}{}{\input{./figures/trace-partial-i.tikz}}
\endpgfgraphicnamed \;\right) \;=\; %
\beginpgfgraphicnamed{trace-partial}
\InputIfFileExists{trace-partial.tikz}{}{\input{./figures/trace-partial.tikz}}
\endpgfgraphicnamed
\]

By adding duals we have enlarged the class of possible diagrams, since
wires may now loop back from inputs to outputs and vice versa, but the
basic principle of diagram equality does not change: \textbf{Two
  diagrams are considered equal if one can be smoothly transformed to
  another, by bending, stretching, or crossing wires, and moving boxes
  around}.  With this in mind we can update the key theorem.

\begin{theorem}[Fundamental Theorem of Diagrams with Daggers and Duals]
  \label{thm:fund-thm-diags-i}
  Given any two generalised compositional theories \catC and \catD
  with daggers and duals, 
  and a map $F:\catC \to \catD$, for any two diagrams $d$ and $d'$ in
  \catC, if $d = d'$ \emph{as diagrams} then $Fd = Fd'$ in \catD.
\end{theorem}
Once again, the full details are found in \cite{SelingerSurvey}.

\begin{remark}\label{rem:duals-and-functors}
  We need not demand any additional conditions on the class of
  mappings to guarantee the preservation of duals; since they are
  defined in terms of processes, the structure is automatically
  preserved.
\end{remark}

\subsection{Observable Structures}
\label{sec:observ-struct}

An \emph{observable} yields classical data from a physical system
\cite{CPaqPav}. The key distinction between classical and quantum data
is that classical data may be freely copied and deleted, while this is
impossible for quantum data, due to the no-cloning \cite{Dieks, WZ}
and no-deleting \cite{Pati} theorems.

In quantum mechanics, an observable is represented by a self-adjoint
operator.  This (non-degenerate) operator encodes certain classical
data as its orthonormal basis of eigenstates, the possible outcomes of
the corresponding measurement.  Note that if a quantum state is known to
be a member of a given orthonormal basis, such as the eigenbasis
$\{\ket{a_i}\}$ of some observable, then it \emph{can} be copied and
deleted via the maps
\[
\delta : \ket{a_i} \mapsto \ket{a_i} \otimes \ket{a_i}
\quad\text{ and }\quad
\epsilon : \ket{a_i} \mapsto 1.
\]
Hence we can view the classical content of a quantum measurement as
the possibility to copy and delete its set of outcomes.  We will
axiomatise quantum observables by describing the copying and deleting
operations as algebraic structures inside a general compositional
theory.  The relevant structure is called a \emph{$\dagger$-special
  commutative Frobenius algebra}, and we will now build up its definition
one piece at a time.

\begin{definition}\label{def:monoid}
  A \emph{commutative monoid} in $\catC$ is a triple $(X, \mu ,
  \eta)$, where $\mu$ and $\eta$ are maps
  \[
    \mu : X \otimes X  \rightarrow X\qquad\qquad\qquad
    \eta :I \rightarrow X
    \]
    which we write graphically as $\mu = \whitemult$, $\eta =
    \whiteunit$.  These operations satisfy the following equations:
  \[
\beginpgfgraphicnamed{associativity}
\InputIfFileExists{associativity.tikz}{}{\input{./figures/associativity.tikz}}
\endpgfgraphicnamed,
  \qquad\quad  %
\beginpgfgraphicnamed{unit}
\InputIfFileExists{unit.tikz}{}{\input{./figures/unit.tikz}}
\endpgfgraphicnamed,
  \qquad\quad  %
\beginpgfgraphicnamed{commutativity}
\InputIfFileExists{commutativity.tikz}{}{\input{./figures/commutativity.tikz}}
\endpgfgraphicnamed\;\;.
  \]
\end{definition}

\begin{remark}\label{rem:monoid-explain}
  The process $\mu$ can be understood as a \emph{multiplication} for
  systems of type $X$; the first and last equations assert that this
  operation is associative and commutative respectively.  The process
  $\eta$ is the \emph{unit} for this multiplication: the second
  equation asserts that multiplication by the unit is simply the
  identity.
\end{remark}

The dual to a monoid is a \emph{comonoid}.

\begin{definition}\label{def:comonoid}
  A \emph{comonoid} in a theory \catC consists of a triple
  $(X,\delta,\epsilon)$ where $\delta$ and $\epsilon$ are processes
  \[
  \delta:X \to X \otimes X \qquad\qquad\qquad\epsilon:X \to I
  \]
  satisfying the equations of Definition~\ref{def:monoid} but in
  reverse, viz:
  \[
\beginpgfgraphicnamed{coassoc}
\InputIfFileExists{coassoc.tikz}{}{\input{./figures/coassoc.tikz}}
\endpgfgraphicnamed
  \qquad\qquad
\beginpgfgraphicnamed{counit}
\InputIfFileExists{counit.tikz}{}{\input{./figures/counit.tikz}}
\endpgfgraphicnamed
  \]
  A comonoid is \emph{cocommutative} if it satisfies:
  \[
\beginpgfgraphicnamed{cocomm}
\InputIfFileExists{cocomm.tikz}{}{\input{./figures/cocomm.tikz}}
\endpgfgraphicnamed.
  \]
\end{definition}
The processes $\delta$ and $\epsilon$ are called the
\emph{comultiplication} and \emph{counit} respectively.

\begin{example}\label{ex:fdhilb-comonoid}
  We have already met the basic example of a comonoid: in \fdhilb, for
  any orthonormal basis $\{x_i\}_i$ of a space $X$ we obtain a
  comonoid via `copying' and `erasing' processes mentioned above:
  \[
  \delta : x_i \mapsto x_i \otimes x_i \qquad \epsilon : x_i \mapsto 1
  \]
\end{example}

\begin{remark}
  Thanks to the dagger, if $(X,\delta,\epsilon)$ is a comonoid then
  $(X,\delta^\dagger,\epsilon^\dagger)$ is automatically a monoid, and
  vice versa.
\end{remark}

Generally speaking, a process is called a \emph{homomorphism} if it
preserves some algebraic structure.  In the context of GCTs, such
preservation is usually expressed by a process commuting with another
which reifies that structure.  For example:

\begin{definition}\label{def:comonoid-homomorphism}
  Given two comonoids $(X,\delta,\epsilon)$ and
  $(X',\delta',\epsilon')$, a \emph{comonoid 
    homomorphism} is a process $f:X\to X'$ such that 
  \[
  \delta' \circ f = (f \otimes f) \circ \delta
  \quad\text{ and }\quad
  \epsilon' \circ f = \epsilon\;.
  \]
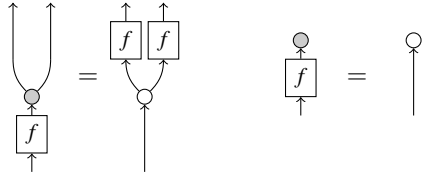
  \[
\beginpgfgraphicnamed{comonoid-homo-i}
\InputIfFileExists{comonoid-homo-i.tikz}{}{\input{./figures/comonoid-homo-i.tikz}}
\endpgfgraphicnamed
  \qquad\qquad
\beginpgfgraphicnamed{comonoid-homo-ii}
\InputIfFileExists{comonoid-homo-ii.tikz}{}{\input{./figures/comonoid-homo-ii.tikz}}
\endpgfgraphicnamed
  \]
Monoid homomorphisms are defined similarly.
\end{definition}

\begin{remark}
  The definition above is the most general, but we will frequently
  encounter cases where $f:X\to X$ is homomorphism between two comonoids
  defined on the same object, or from a single comonoid to itself.
\end{remark}

The structures of greatest interest for this paper are algebras
containing both monoids and comonoids.

\begin{definition}\label{def:frobenius-alg}
  A \emph{commutative Frobenius algebra} is a 5-tuple
  $(X,\delta,\epsilon,\mu,\eta)$ where 
  \begin{enumerate}
  \item $(X,\delta,\epsilon)$ is a cocommutative comonoid;
  \item $(X,\mu,\eta)$ is a commutative monoid; and,
  \item $\delta$ and $\mu$ satisfy the following equations:
    \[
\beginpgfgraphicnamed{frob-condition}
\InputIfFileExists{frob-condition.tikz}{}{\input{./figures/frob-condition.tikz}}
\endpgfgraphicnamed
    \] 
  \end{enumerate}
\end{definition}

Finally, we can define:

\begin{definition}\label{def:dscfa}
  A \emph{$\dagger$-special Frobenius algebra} ($\dagger$-SCFA) is a commutative Frobenius algebra
  \begin{align*}
  \mathcal O_{\!\smallwhitedot} = ( &
    \whitemu : X \otimes X \rightarrow X, \ \ 
    \whiteeta : I \rightarrow X, \\
  & \whitedelta : X \rightarrow X \otimes X, \ \ 
    \whiteepsilon : X \rightarrow I)
  \end{align*}
  such that $\whitedelta = (\whitemu)^\dagger$, $\whiteepsilon = (\whiteeta)^\dagger$ and
\beginpgfgraphicnamed{special}
\InputIfFileExists{special.tikz}{}{\input{./figures/special.tikz}}
\endpgfgraphicnamed.
\end{definition}

The preceding definitions may seem rather opaque, and not fully
justified by the intuition that a quantum observable is somehow
encoded by the maps which copy and delete its eigenstates.  However
complex it may appear (and we shall shortly simplify it), the
importance of Definition~\ref{def:dscfa} rests on the fact
\cite{CPV2008} that in \fdhilb \emph{every} 
$\dagger$-SCFA arises from a comonoid defined by copying an
orthonormal basis as described above.  Since orthonormal bases define
non-degenerate quantum observables, $\dagger$-SCFAs are also called
\emph{observable structures}.

Concretely, given an orthonormal basis $\{\ket{i}\}_i$ then
$\whitedelta:: \ket{i}\mapsto \ket{ii}$ defines an observable, and
all observables are of this form for some orthonormal basis.  The resulting
intuition is that $\whitedelta$ is an operation that  `copies' basis
vectors, and  that $\whiteepsilon$  `erases' them \cite{CPaqPav}. 
 We will use the symbolic representation
$(\whitemu, \whiteeta, \whitedelta, \whiteepsilon)$ and the pictorial
one $(\whitemult, \whiteunit, \whitecomult, \whitecounit)$
interchangeably. 

\begin{example}[Sets and Relations]\label{ex:observables-in-Rel}
  Perhaps surprisingly, ${\bf FRel}$ also has many distinct observable
  structures, which have been classified by Pavlovic \cite{Dusko}.
  Even on the two element set there are two, namely
\begin{align*}
\whitedelta: & \; i \mapsto \{(i, i)\}
\\
\graydelta: & \;\left\{
\begin{array}{lcr}
0 & \mapsto & \{(0,0), (1,1)\}  \\
1 & \mapsto & \{(0,1), (1,0)\}
\end{array}\right.
\end{align*}
In fact, this pair is \emph{strongly complementary} in the sense to be
explained in Section~\ref{Sec:complementarity}.
\end{example}

Manipulating observable structures in the graphical language is
extremely convenient since they obey a remarkable normal form law.  
Let $\delta_n: X \to X^{\otimes n}$ be defined by the recursion
\[
\delta_0 := \epsilon 
\qquad 
\delta_{n+1} := (\delta_n \otimes 1_A) \circ \delta
\]
and define $\mu_m$ analogously.  Now we have the following important
theorem:
\begin{theorem}\label{thm:frob-alg-nf}
  Given an SCFA $(X,\delta,\epsilon,\mu,\eta)$ let $f:X^{\otimes m}\to
  X^{\otimes n}$ be a map constructed from $\delta$,$\epsilon$,$\mu$
  and $\eta$ whose graphical form is connected.  Then $f =
  \delta_n \circ \mu_m$.
\end{theorem}

\begin{proposition}\label{prop:spider}\em 
Given an observable structure $\whiteobs$ on $X$, let
$(\whitedot)_n^m$ denote the `$(n,m)$-legged spider':
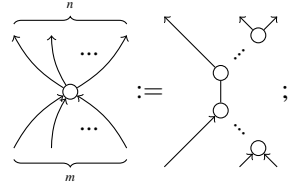
\[
 \begin{tikzpicture}
    \begin{pgfonlayer}{nodelayer}
        \node [style=none] (0) at (-2, 1) {};
        \node [style=none] (1) at (-0.5, 1) {};
        \node [style=none] (2) at (0, 1) {};
        \node [style=none] (3) at (1, 1) {};
        \node [style=none] (4) at (1.5, 1) {};
        \node [style=none] (5) at (-2, 0.75) {};
        \node [style=none] (6) at (-1.5, 0.75) {};
        \node [style=none] (7) at (-0.5, 0.75) {};
        \node [style=white dot] (8) at (1.25, 0.75) {};
        \node [style=none] (9) at (-1, 0.5) {...};
        \node [style=whitebg] (10) at (1, 0.5) {\small \rotatebox[origin=c]{45}{...}};
        \node [style=white dot] (11) at (0.75, 0.25) {};
        \node [style=white dot] (12) at (-1.25, 0) {};
        \node [style=none, anchor=east] (13) at (0, 0) {$:=$};
        \node [style=white dot] (14) at (0.75, -0.25) {};
        \node [style=none] (15) at (-1, -0.5) {...};
        \node [style=whitebg, fill=white] (16) at (1, -0.5) {\small \rotatebox[origin=c]{-45}{...}};
        \node [style=none] (17) at (-2, -0.75) {};
        \node [style=none] (18) at (-1.5, -0.75) {};
        \node [style=none] (19) at (-0.5, -0.75) {};
        \node [style=white dot] (20) at (1.25, -0.75) {};
        \node [style=none] (21) at (-2, -1) {};
        \node [style=none] (22) at (-0.5, -1) {};
        \node [style=none] (23) at (0, -1) {};
        \node [style=none] (24) at (1, -1) {};
        \node [style=none] (25) at (1.5, -1) {};
    \end{pgfonlayer}
    \begin{pgfonlayer}{edgelayer}
        \draw [style=diredge, bend left=15] (18.center) to (12);
        \draw (14) to (11);
        \draw [style=diredge] (11) to (2.center);
        \draw [style=small braceedge] (22.center) to node[wire label, inner sep=3 pt]{$m$} (21.center);
        \draw [style=diredge] (8) to (4.center);
        \draw [style=diredge] (23.center) to (14);
        \draw [style=diredge] (24.center) to (20);
        \draw [style=diredge] (8) to (3.center);
        \draw [style=diredge, bend left=15] (12) to (5.center);
        \draw [style=diredge, bend right=15] (19.center) to (12);
        \draw [style=diredge, bend left=15] (17.center) to (12);
        \draw [style=diredge, bend right=15] (12) to (7.center);
        \draw [style=diredge] (25.center) to (20);
        \draw [style=diredge, bend left=15] (12) to (6.center);
        \draw [style=small braceedge] (0.center) to node[wire label, inner sep=3 pt]{$n$} (1.center);
    \end{pgfonlayer}
\end{tikzpicture} \,;
\]
then any morphism $X^{\otimes n}\to X^{\otimes m}$ built from
$\whitemu, \whiteeta, \whitedelta$ and $\whiteepsilon$ via $\dagger$-SMC structure which has a
connected graph is equal to the $(\whitedot)^m_n$. Hence, spiders
compose as follows:
\begin{equation}\label{eq:spidercomp}
\beginpgfgraphicnamed{spidercomp}
\InputIfFileExists{spidercomp.tikz}{}{\input{./figures/spidercomp.tikz}}
\endpgfgraphicnamed
\end{equation}
\end{proposition}

Thanks to the spider rule (\ref{eq:spidercomp}), every observable
structure on $X$ makes $X$ dual to itself (in the sense of
Definition~\ref{def:duals}), via the cup and cap:
\ctikzfig{frob_comp}
The upper-star with respect to this cup and cap corresponds in
\catFHilb to transposition in the given basis. For that reason, we
call this the $\whitedot$-transpose $f^\whitetranspose$. The lower
star corresponds to complex conjugation in the basis of $\whiteobs$,
so we call it the $\whitedot$-conjugate $f_\whiteconjugate :=
(f^\whitetranspose)^\dagger$.

Recall that a process $k:I\to X$ is called a \emph{point of $X$}.  In
\fdhilb the points of $X$ are simply vectors in the Hilbert space $X$.
The abstract analogue of the eigenvectors of an observable in \fdhilb
are the \emph{classical points} of an observable structure.

\begin{definition}
  \label{def:classicalpoint}
  A \emph{classical point} for an observable structure is a state that
  is copied by the comultiplication and deleted by the counit:
\begin{equation}\label{eq:classpoint}
\beginpgfgraphicnamed{classpoint}
\InputIfFileExists{classpoint.tikz}{}{\input{./figures/classpoint.tikz}}
\endpgfgraphicnamed 
\end{equation}
\end{definition}
We will depict classical points as triangles of the same colour as
their observable structure. 

\begin{remark}\label{rem:classpoints-are-comonoid-homos}
  Another way to say the same thing, is to define classical points
  as comonoid homomorphisms from the trivial comonoid $(I,1_I,1_I)$
  to $(X,\delta,\epsilon)$.
\end{remark}

In quantum computing, it is common to think of elements of a product
basis as strings of some kind. E.g. for qubits:
\[ \ket{010011} \leftrightarrow
   \ket 0 \otimes \ket 1 \otimes \ket 0 \otimes
   \ket 0 \otimes \ket 1 \otimes \ket 1
\]
Such product bases are precisely the classical points of
\emph{products of observable structures}.  Given an observable
structure \whiteobs on $X$, and another \grayobs on $Y$, we form a new
observable structure on $X \otimes Y$ by taking their product:
\[
\delta = %
\beginpgfgraphicnamed{product-delta}
\InputIfFileExists{product-delta.tikz}{}{\input{./figures/product-delta.tikz}}
\endpgfgraphicnamed \qquad \epsilon = %
\beginpgfgraphicnamed{product-epsilon}
\InputIfFileExists{product-epsilon.tikz}{}{\input{./figures/product-epsilon.tikz}}
\endpgfgraphicnamed
\]
Evidently any pair of classical points for the constituent observable
structures will be copied.  
\ctikzfig{product-class-points}
Generalising, the classical points of any
$n$-ary product of observable structures are nothing more than lists
of classical points, one for each factor.

Working concretely in Hilbert space, one can use the linear structure
to give another set of equations for observable structures. Consider
some basis vector $\ket{i}$, then the map $\outp{ii}{i}$ has the diagrammatic
form: 
\ctikzfig{presumdelta}
But notice that the sum $\sum_i \outp{ii}{i}$ is nothing more than the
the map $\delta: \ket{i} \to \ket{ii}$.  A similar statement can be
made for the counit $\epsilon$.  Hence given the complete set
of classical points for an observable structure \grayobs we have the
following equations:
  \ctikzfig{spider_sums}
Indeed these can be generalised to arbitrary spiders:
  \ctikzfig{spider_sums_n_legs}
Note that generalised compositional theories do not necessarily admit
addition of diagrams:  we introduce these equations as way of
generalising from concepts defined in Hilbert space to the abstract
setting where there need not be any linear structure.

Linear maps have the property that if two maps are equal on every
element of a basis, the maps themselves are equal.  In analogy to this
we define the following:
\begin{definition}
  \label{def:enough-classical-points}
  Let \whiteobs be an observable structure on $X$, with classical
  points $\{k_i\}_i$; we say that \whiteobs has \emph{enough classical
    points} if, for every system $Y$, and every pair of processes $f,g:
  X\to Y$, we have
  \[
  ( \forall i : f\circ k_i = g\circ k_i) \Rightarrow f = g\;\;.
  \]
\end{definition}
This property does not necessarily hold in an arbitrary GCT (although
obviously it does in \fdhilb) however when it does hold certain
statements can be made stronger.  For example, many implications
described in the subsequent sections are equivalences if the
underlying object has enough classical points.

\subsection{Phase Group for an Observable Structure}\label{sec:phasegroup}
 
Let $\psi$ and $\phi$ be two points of $X$.  Given an observable
structure $\whiteobs$ on $X$, applying the 
multiplication $\whitemu$ to $\psi$ and $\phi$ yields another point of
$X$: 
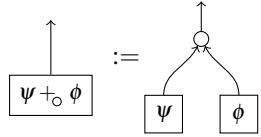
\begin{gather*}
  \psi \whiteplus \phi \;:=\; \whitemu(\psi \otimes \phi)\\\\
\beginpgfgraphicnamed{whiteplus-def}
\InputIfFileExists{whiteplus-def.tikz}{}{\input{./figures/whiteplus-def.tikz}}
\endpgfgraphicnamed
\end{gather*}
Since $\whitemu$ is commutative and associative, and it has a unit
point (namely $\whiteeta$), the operation $\whiteplus$ gives the
points of $X$ the structure of a commutative monoid.

If we restrict to those points $\psi : I \to X$ which satisfy 
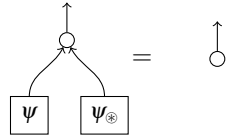
\begin{gather*}
  \psi \whiteplus \psi_\whiteconjugate = \whiteeta
  \\\\
\beginpgfgraphicnamed{whiteplus-group}
\InputIfFileExists{whiteplus-group.tikz}{}{\input{./figures/whiteplus-group.tikz}}
\endpgfgraphicnamed
\end{gather*}
we obtain an abelian group $\whitePhi$, called the \emph{phase group}
of $\whiteobs$~\cite{CD2008,CD2009}. The elements of this group are
called \emph{phases}.  Note the phase group is non-empty, since the
unit \whiteeta satisfies the required equation.  We let $-\alpha :=
\alpha_\whiteconjugate$ for phases, and represent these points as
circles with one output, labelled by a phase.
\ctikzfig{phases}
\begin{example}\label{ex:phasegroup-in-fHilb}
  In \fdhilb, let  \whiteobs be defined by some
  orthonormal basis $\{\ket{i}\}_i$.  One can verify by direct
  calculation that a vector $\ket\psi$ lies in the phase group
  \whitePhi if and only if we have $|\innp{i}{\psi}|^2 = 1/D$, for all $i$,
  where $D$ is the dimension of the underlying space.  Such vectors
  are called \emph{unbiased} for the basis $\{\ket{i}\}_i$.  The
  multiplication is then the convolution (pointwise) product.

Concretely, for a qubit observable  given by $\whitemu=
\ketbra{0}{00} + \ketbra{1}{11}$, the phases are the unbiased states,
which are all of the form:
\[
\ket{\alpha} =
\begin{pmatrix}
  1 \\ e^{i\alpha}
\end{pmatrix}\;,
\]
with the multiplication:
\[
\whitemu
\left(
  \begin{pmatrix}
    1\\ e^{i\alpha}
  \end{pmatrix}
\otimes
\begin{pmatrix}
  1\\ e^{i\beta}
\end{pmatrix}
\right)
=
\begin{pmatrix}
  1\\ e^{i(\alpha+\beta)}
\end{pmatrix}\;.
\] 
We therefore see that the phase group for a qubit observable is the
circle group.  It is an easy exercise to check that for a
$D$-dimensional Hilbert space the phase group for any observable is
isomorphic to the $(D-1)$-dimensional torus.
\end{example}

The name `phase group' comes from fact that the elements of the
\whitePhi correspond to unitary maps which preserve the basis defining
\whiteobs.  We can map any point $\psi:I\to X$ onto an operation on
$X$ via the \emph{left action}, $\whiteLambda (\psi) = \whitemu \circ
(\psi \otimes 1_X)$, or in pictures:
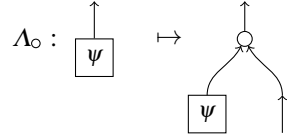
\[
\whiteLambda : \;\;%
\beginpgfgraphicnamed{whitelambda-def}
\InputIfFileExists{whitelambda-def.tikz}{}{\input{./figures/whitelambda-def.tikz}}
\endpgfgraphicnamed
\]
Then we have the following facts:

\begin{proposition}\label{thm:phase-group-props}
  Let $\phi,\psi \in \whitePhi$; then 
  \begin{enumerate}
  \item $\whiteLambda(\phi)$ is unitary;
  \item $\whiteLambda(\phi)\circ \whiteLambda(\psi) 
    = \whiteLambda(\phi+\psi)
    = \whiteLambda(\psi)\circ \whiteLambda(\phi) $
  \item If $k$ is a classical point for \whiteobs then
    $\whiteLambda(\phi)\circ k = k \otimes s$ for some
    scalar $s$.
  \end{enumerate}
  \begin{proof}
    \begin{enumerate}
    \item We show that $\whiteLambda(\psi)^\dagger \circ
      \whiteLambda(\psi) = 1$: \ctikzfig{whitelambda-unitary} The
      first equation is the spider rule while the second is the
      definition of $\psi_\whiteconjugate$.  The case
      $\whiteLambda(\psi) \circ (\whiteLambda(\psi))^\dagger = 1$ is
      similar.
      \item This follows immediately from the associativity and
        commutativity of $\whitemu$:
        \ctikzfig{whitelambda-plus}
      \item This follow from the definition of classical points.
        \ctikzfig{whitelambda-classical}
    \end{enumerate}
  \end{proof}
\end{proposition}

The image $\whiteLambda(\whitePhi)$ is therefore an abelian subgroup
of the unitaries on $X$, which is isomorphic to $\whitePhi$.  We refer
to these as \emph{phase maps}.  If we reinterpret the third part of
the preceding proposition in terms of linear algebra, we see that every
classical point of \whiteobs is an eigenvector of every phase map in 
$\whiteLambda(\whitePhi)$.  This in turn ``explains'' why they commute
with each other.

\begin{example}\label{ex:phasegroup-maps-in-fdhilb}
  Let \whiteobs be defined by $\whitemu=
\ketbra{0}{00} + \ketbra{1}{11}$ as above.  Now for $\alpha \in
\whitePhi$ we have 
\[
\whiteLambda(\alpha) = 
\begin{pmatrix*}[c]
  1 & 0 \\ 0 & e^{i\alpha}
\end{pmatrix*}
\]
Hence the phase group in Hilbert spaces is exactly the group of phase
shifts relative to the given basis.
\end{example}



Generalising from the preceding discussion, we can now introduce
`spiders decorated with phases':
\begin{equation}\label{eq:decspider}
\beginpgfgraphicnamed{decspider}
\InputIfFileExists{decspider.tikz}{}{\input{./figures/decspider.tikz}}
\endpgfgraphicnamed
\end{equation}
which compose as follows:
\begin{equation}\label{eq:decspidercomp} 
\beginpgfgraphicnamed{decspidercomp}
\InputIfFileExists{decspidercomp.tikz}{}{\input{./figures/decspidercomp.tikz}}
\endpgfgraphicnamed
\end{equation}
In the following sections we will refer to this generalised
composition rule for phased spiders as \emph{the spider law}.

\subsection{Two toy models}
\label{sec:two-toy-models}

In this section we'll introduce two ``toy'' models of quantum
mechanics.  The first is the restriction of quantum mechanics to
stabilizer states;  this theory we call \stab.  The second is the toy
model due to Spekkens \cite{Spekkens}, which we refer to as \spek.
While the first of these is indeed a true subtheory of quantum
mechanics, \spek is a local hidden variable model.  By casting both of
these in the language of generalised compositional theories we can see
that the difference between is in fact very slight.

Before discussing these concrete models, we'll introduce a formal
precursor theory.  Let \toy be the general compositional theory built
from the formal generators:
\begin{itemize}
\item one basic system, which we denote $T$;
\item six points $z_0,z_1,x_0,x_1,y_0,y_1 : I \to T$ and their
  corresponding copoints;
\item 24 unitary maps $T\to T$ which form a group isomorphic to the
  symmetric group $S_4$;
\item one observable structure \whiteobs, whose classical points are
  $z_0$ and $z_1$, and whose phase group comprises the remaining four
  points.
\end{itemize}
Note that \toy is not fully specified: to do so we ought to say which
group the phase group is, and how the corresponding unitaries embed
into the endomorphisms of $T$.  Since \whitePhi is a four-element
group we have only two choices here: $\mathbb{Z}_4$, or $\mathbb{Z}_2
\times \mathbb{Z}_2$.  As we will see this choice will make the
difference between stabilizer quantum mechanics and the quantum-like
local hidden variable theory.

Let \stab be the subtheory of \fdhilb generated by the following
elements:
\begin{itemize}
\item One basic system $\mathbb{C}^2$, which we call $Q$.
\item Six points $I\to Q$:
  \begin{align*}
z_0 &= \ket0  & 
x_0 & = \frac{1}{\sqrt 2}(\ket0 + \ket1) & 
y_0 & = \frac{1}{\sqrt 2}(\ket0 + i\ket1) \\
z_1 &= \ket1 & 
x_1 &= \frac{1}{\sqrt 2}(\ket0 - \ket1) &
y_1 &= \frac{1}{\sqrt 2}(\ket0 - i\ket1)
  \end{align*}
\item The group of unitaries generated by the matrices:
\[
Z_{\pi/2} = 
\begin{pmatrix*}
  1 & 0 \\ 0 & i
\end{pmatrix*}
\qquad
X_{\pi/2} = 
\frac{1}{\sqrt{-2i}}
\begin{pmatrix*}[r]
1 & -i \\ -i & 1  
\end{pmatrix*}
\]
This group is known in the quantum computation literature as the Clifford
group for one qubit;  it is isomorphic to $S_4$.  The other key
property of this group is that it acts as a permutation on the states
defined above, so we cannot generate new states via unitaries.
\item An observable structure \whiteobs defined by the basis $\ket0,\ket1$.
\end{itemize}
Evidently the classical points of \whiteobs are indeed
$z_0$ and $z_1$ and the remaining points are unbiased for this basis,
hence part of \whitePhi.  One can check that
\[
\whiteLambda(y_0) = Z_{\pi/2}
\]
and which generates a four element cyclic subgroup, hence the phase group
\whitePhi is $\mathbb{Z}_4$.

We now introduce Spekkens' toy theory.  The toy theory is a local
hidden variable theory, based on epistemic restrictions.  There is a
single basic system, the toy bit, which can have one of four possible
states. We formalise the state space simply as a four-element set.
However, we now impose the epistemic restriction that any state
preparation (and, dually, measurement) may only narrow down the state
to two of the possible four.  Hence the ``states'' of the toy theory
are two-element subsets.  Although Spekkens' original presentation
\cite{Spekkens} was informal, the toy theory is ideally studied as
subtheory of \frel.  Following \cite{EdwardsSpek2011}, let \spek be
the subtheory of \frel generated by the following elements:
\begin{itemize}
\item One basic system, the four element set $\mathbf{4} = \{0,1,2,3\}$.
\item Six points:
  \begin{align*}
    z_0 &= \{0,1\}  & 
    x_0 &= \{0,2\}  &  
    y_0 &= \{0,3\}  \\
    z_1 &= \{2,3\}  & 
    x_1 &= \{1,3\}  &  
    y_1 &= \{1,2\} 
  \end{align*}
\item The full group of permutations on $\mathbf{4}$;
\item An observable structure \whiteobs defined by
\[
\whitemu : 
\begin{array}{rcl}
 \{00,11\} &  \sim & 0 \\
 \{01,10\} &  \sim & 1 \\
 \{22,33\} &  \sim & 2 \\
 \{23,32\} &  \sim & 3  
\end{array}
\qquad \whiteeta : * \sim \{0,2\}
\]
where we write the tensor as juxtaposition, i.e. $00 = (0,0)$.
\end{itemize}
Once again we easily check that the classical points for \whiteobs are
$z_0$ and $z_1$, and the other four form the phase group \whitePhi.
The phase group in this case is generated by the transpositions $(0
\;1)$ and $(2\;3)$; hence $\whitePhi \iso \mathbb{Z}_2
\times\mathbb{Z}_2$.

As should be evident by this point both \stab and \spek are
realisations of the incomplete theory \toy.  The only notable
difference between them is the group structure of \whitePhi.  This
highlights the importance of the  phase group for understanding
non-locality in generalised compositional theories.

\begin{remark}\label{rem:completing-stab-spek}
  In the description above the group of unitaries was given \emph{a
    priori}.  This is not necessary.  If we include a second
  observable structure \grayobs, corresponding to the classical points $x_0$
  and $x_1$, the the union of the two phase groups \whitePhi and \grayPhi yields
  \emph{all} unitaries described above.  These two observables are
  complementary in the sense described below.  Hence these two
  theories are in a sense minimal.
\end{remark}

\section{Complementarity and Strong Complementarity}\label{Sec:complementarity}


In the Hilbert space presentation of quantum mechanics, two observables are \emph{complementary} if their bases of eigenstates are mutually unbiased. That is, for any $i,j$, $|\braket{v_i}{v_j'}|^2 = 1/D$. In the graphical notation:
\ctikzfig{mub}

A question posed by Coecke and Duncan \cite{CD2008,CD2009} was, ``Can we represent complementarity purely in terms of interacting observable structures?'' It turns out that complementarity is equivalent to a simple diagrammatic equation. First, we can move $1/D$ in the above equation to the other side and express it as a circle, as the trace of the identity always equals $D$. Then, replace $1$ on the RHS with ``deleted points''.
\begin{equation}\label{eq:mub2}
\beginpgfgraphicnamed{mub2}
\InputIfFileExists{mub2.tikz}{}{\input{./figures/mub2.tikz}}
\endpgfgraphicnamed
\end{equation}
As we saw in section~\ref{sec:observ-struct}, observable structures
fix an isomorphism of a space with its dual space, via the
transpose. While it is not true in general that
$\ket\psi^\whitetranspose = \bra\psi$, the transpose does take
classical points for a particular observable structure to their
adjoints: 
\[
\ket{v_i}^\whitetranspose = \bra{v_i} 
\quad\text{ and }\quad
\ket{v_j'}^\graytranspose = \bra{v_j'}\;.
\] 
Graphically:
\begin{equation}
\beginpgfgraphicnamed{dag_frob_transpose}
\InputIfFileExists{dag_frob_transpose.tikz}{}{\input{./figures/dag_frob_transpose.tikz}}
\endpgfgraphicnamed\label{eq:class-point-transpose}
\end{equation}

\begin{exercise*}
  Prove this.
\end{exercise*}

\noindent
We can rewrite the left hand side of equation (\ref{eq:mub2}) using this fact.
\begin{equation}\label{eq:complmentary-bis}
\beginpgfgraphicnamed{mub3}
\InputIfFileExists{mub3.tikz}{}{\input{./figures/mub3.tikz}}
\endpgfgraphicnamed
\end{equation}
The last equation follows by substituting the symbol $S$ for its
definition, viz:
\begin{equation}\label{eq:mub-antipode}
\beginpgfgraphicnamed{mub_antipode}
\InputIfFileExists{mub_antipode.tikz}{}{\input{./figures/mub_antipode.tikz}}
\endpgfgraphicnamed
\end{equation}
Unifying Equations~\eqref{eq:mub2} and \eqref{eq:complmentary-bis} we have:
\begin{equation}\label{eq:pre-hopf-antipode}
\beginpgfgraphicnamed{mub4}
\InputIfFileExists{mub4.tikz}{}{\input{./figures/mub4.tikz}}
\endpgfgraphicnamed
\end{equation}
Since this equation holds for all classical points $i$ and $j$, if we
now appeal to the fact that \fdhilb has enough classical points
(cf. Definition~\ref{def:enough-classical-points}), we can conclude
that identity holds without points:
\begin{equation}\label{eq:antipode-hopf}
\beginpgfgraphicnamed{antipode_hopf}
\InputIfFileExists{antipode_hopf.tikz}{}{\input{./figures/antipode_hopf.tikz}}
\endpgfgraphicnamed
\end{equation}

\begin{remark}\label{rem:why-antipode}
  The above equation is (up to a scalar factor) one of the defining
  equations a \emph{Hopf algebra}, in which case the map $S$ is called
  the \emph{antipode}.  For that reason, we refer to
  (\ref{eq:antipode-hopf}) as the \textit{Hopf law}.  As we will see
  in the next section, subject to some additional assumptions, pairs
  of complementarity observables do indeed form Hopf algebras with the
  antipode defined as in Equation \eqref{eq:mub-antipode}.
\end{remark}

Notice if we assume Equation (\ref{eq:antipode-hopf}) we can derive
Equation \eqref{eq:mub2} without any additional assumptions.  In other words,
if $\whiteobs$ and $\grayobs$ satisfy the Hopf law their classical
points are mutually unbiased.  Thus, we take the Hopf law to be the
defining equation for our \emph{abstract} notion of complementarity.


\begin{definition}\label{def:complementary-obs}
A pair $(\whiteobs, \grayobs)$ of observables on the same object
$X$ is \emph{complementary} iff:
\[
\beginpgfgraphicnamed{antipode_hopf}
\InputIfFileExists{antipode_hopf.tikz}{}{\input{./figures/antipode_hopf.tikz}}
\endpgfgraphicnamed \quad\ \mbox{\rm where} \ \ \quad %
\beginpgfgraphicnamed{mub_antipode}
\InputIfFileExists{mub_antipode.tikz}{}{\input{./figures/mub_antipode.tikz}}
\endpgfgraphicnamed\ .
\]
\end{definition}

Since every observable in \fdhilb has enough classical points,
Definition~\ref{def:complementary-obs} is equivalent\footnote{Indeed
  Equations~\eqref{eq:mub2} and \eqref{eq:antipode-hopf} are
  equivalent in any theory wherever at least one of the observable
  structures has enough classical points.}
  to saying that
observables are complementary if their eigenbases are mutually unbiased with
respect to the other.  (See \cite{CD2008} for more details).  Hence,
we reclaim the usual notion of quantum complementarity, and extend it
to a more general setting.

\begin{definition}\label{def:coherence}
A pair $(\whiteobs, \grayobs)$ of observables on the same object
$X$ is \emph{coherent} iff: 
\[
\beginpgfgraphicnamed{coher}
\InputIfFileExists{coher.tikz}{}{\input{./figures/coher.tikz}}
\endpgfgraphicnamed\ .
\]
\end{definition}
In other words, the unit point $\whiteeta$ (\whiteunit) is, modulo a
scalar factor, a classical point for $\grayobs$, and vice versa.
\[
\exists i, j : %
\beginpgfgraphicnamed{classpoint-cohere}
\InputIfFileExists{classpoint-cohere.tikz}{}{\input{./figures/classpoint-cohere.tikz}}
\endpgfgraphicnamed
\]
We will assume that the scalar $%
\beginpgfgraphicnamed{scalar}
\begin{tikzpicture}
	\begin{pgfonlayer}{nodelayer}
		\node [style=gray dot] (0) at (0, 0.12) {};
		\node [style=white dot] (1) at (0, -0.12) {};
	\end{pgfonlayer}
	\begin{pgfonlayer}{edgelayer}
		\draw [style=none] (1) to (0);
	\end{pgfonlayer}
\end{tikzpicture}}
\endpgfgraphicnamed$ is always
cancellable.

\begin{example}\label{ex:coherence}
  Consider the two observables on the Hilbert space $\mathbb{C}^2$
  corresponding to the $Z$ and $X$ spins:
  \begin{align*}
    \whitedelta : &
      \begin{array}{rcl}
        \ket{0} & \mapsto & \ket{00} \\
        \ket{1} & \mapsto & \ket{11}
      \end{array}
&
    \graydelta : &
      \begin{array}{rcl}
        \ket{+} & \mapsto & \ket{++} \\
        \ket{-} & \mapsto & \ket{--}
      \end{array}
\\ \\
    \whiteeta : &
      \begin{array}{rcl}
        \ket{0} & \mapsto & 1 \\
        \ket{1} & \mapsto & 1
      \end{array}
&
    \grayeta : &
      \begin{array}{rcl}
        \ket{+} & \mapsto & 1 \\
        \ket{-} & \mapsto & 1
      \end{array}
  \end{align*}
Computing $\whiteeta$ we obtain:
\[
\whiteeta = (\whiteepsilon)^\dagger = \Big[ 1 \mapsto (\ket0 + \ket1)
  \Big]=
\sqrt2 \ket+
\]
which is indeed proportional to a classical point for $\graydelta$.
By a similar computation we obtain $\grayeta = \sqrt{2}\ket0$, from
which the value of their inner product $%
\beginpgfgraphicnamed{scalar}
\begin{tikzpicture}
	\begin{pgfonlayer}{nodelayer}
		\node [style=gray dot] (0) at (0, 0.12) {};
		\node [style=white dot] (1) at (0, -0.12) {};
	\end{pgfonlayer}
	\begin{pgfonlayer}{edgelayer}
		\draw [style=none] (1) to (0);
	\end{pgfonlayer}
\end{tikzpicture}}
\endpgfgraphicnamed = \sqrt{2}$
follows.  The equations of definition~\ref{def:coherence} can easily
be verified from here, demonstrating that $\whiteobs$ and $\grayobs$
are coherent.
\end{example}

\begin{proposition}\label{prop:cohere}\em
In \fhilb if $O_{\!\smallwhitedot}$ and $O_{\!\smallgraydot}$
are self-adjoint operators corresponding to complementary observables,
one can always choose a pair of coherent observable structures
$(\whiteobs,\grayobs)$ whose classical points correspond to the
eigenbases of $O_{\!\smallwhitedot}$ and $O_{\!\smallgraydot}$.
\end{proposition}

\begin{proof}
  Let $\{\,\ket{a_i}\,\}_{i=1}^n$ and $\{\,\ket{b_j}\,\}_{j=1}^n$ be
  orthonormal eigenbases for $O_{\!\smallwhitedot}$ and
  $O_{\!\smallgraydot}$ respectively.  Since the bases are mutually
  unbiased we have
  \[
  \ket{b_j} = \frac{1}{\sqrt{n}} \big[ 
    \alpha_{1j}\ket{a_1} + \cdots     \alpha_{nj}\ket{a_n} \big] 
  \]
  where the $\alpha_{ij}$ are scalars satisfying $\sizeof{\alpha_{ij}}
  = 1$.  Setting $\ket{a_i'} = \alpha_{i1}\ket{a_i}$, we see that
  $\{\ket{a'_i}\}_i$ is also an orthonormal eigenbasis for
  $O_{\!\smallwhitedot}$, which is still mutually unbiased with
  respect to $\{\ket{b_j}\}_j$.  Now define:
  \begin{align*}
    \whitedelta : & \ket{a_i'} \mapsto \ket{a_i'} \otimes \ket{a_i'}\\
    \whiteepsilon : & \ket{a_i'} \mapsto 1
  \end{align*}
  This choice yields $\whiteeta = (\whiteepsilon)^\dagger = \sum_i
  \ket{a_i'} = \sqrt{n} \ket{b_1}$.

  In a similar fashion we can choose an eigenbasis $\ket{b_1'}, \ldots
  , \ket{b_n'}$ for
  $O_{\!\smallgraydot}$ such that the resulting $\graydelta$ and
  $\grayepsilon$ satisfy $(\grayepsilon)^\dagger =
  \sqrt{n}\ket{a_1'}$.  It is straightforward to verify that this can
  be done such that $\ket{b_1'} = \ket{b_1}$, ensuring the coherence
  of $\whiteobs$ and $\grayobs$.
\end{proof}

For this reason we will from now on assume that pairs of complementary
observables are always coherent.

\subsection{Strongly complementary observables}
\label{sec:strongly-compl-obser}

Many familiar observables, when expressed in terms of algebras, turn
out to have useful additional properties.  These are called
\emph{strongly complementary}; before describing them we will require
some preliminary definitions.

\begin{definition}\label{def:bialgebra}
  A \emph{(commutative) bialgebra} on $X$ is a 4-tuple
  $(\mu,\eta,\delta,\epsilon)$ of maps,
  \begin{align*}
    \mu : &\; X \otimes X \to X &\; \delta : &\; X \to X \otimes X \\
    \eta : &\; I \to X &\; \epsilon : &\; X \to I\;,
  \end{align*}
  which we write graphically as $(\graymult, \grayunit, \whitecomult,
  \whiteunit)$, such that:
  \begin{itemize}
  \item $(X,\mu,\eta)$ is a (commutative) monoid; 
  \item $(X, \delta, \epsilon)$ is (cocommutative) comonoid; and, 
  \item the following additional equations are satisfied:
    \begin{gather}
\beginpgfgraphicnamed{bialg-std}
\InputIfFileExists{bialg-std.tikz}{}{\input{./figures/bialg-std.tikz}}
\endpgfgraphicnamed\label{eq:bialg-i}  \\
\beginpgfgraphicnamed{bialg-std-ii}
\InputIfFileExists{bialg-std-ii.tikz}{}{\input{./figures/bialg-std-ii.tikz}}
\endpgfgraphicnamed\label{eq:bialg-ii} \\
\beginpgfgraphicnamed{scalar}
\begin{tikzpicture}
	\begin{pgfonlayer}{nodelayer}
		\node [style=gray dot] (0) at (0, 0.12) {};
		\node [style=white dot] (1) at (0, -0.12) {};
	\end{pgfonlayer}
	\begin{pgfonlayer}{edgelayer}
		\draw [style=none] (1) to (0);
	\end{pgfonlayer}
\end{tikzpicture}}
\endpgfgraphicnamed = {} \label{eq:biagl-iii}
    \end{gather}
  \end{itemize}
\end{definition}

\begin{remark}
  Note that equations~\eqref{eq:bialg-ii} and \eqref{eq:biagl-iii} are
  very similar to the equations required for the coherence of two
  observables, per Definition~\ref{def:coherence}.  The only
  difference there is that the scalar %
\beginpgfgraphicnamed{scalar}
\begin{tikzpicture}
	\begin{pgfonlayer}{nodelayer}
		\node [style=gray dot] (0) at (0, 0.12) {};
		\node [style=white dot] (1) at (0, -0.12) {};
	\end{pgfonlayer}
	\begin{pgfonlayer}{edgelayer}
		\draw [style=none] (1) to (0);
	\end{pgfonlayer}
\end{tikzpicture}}
\endpgfgraphicnamed is not assumed to be
  trivial.
\end{remark}

\begin{definition}\label{def:hopf-algebra}
  A \emph{(commutative) Hopf algebra} on $X$ is a (commutative)
  bialgebra on $X$, augmented with a map $s:X\to X$, called the
  \emph{antipode}, which satisfies:
  \begin{equation}    
\beginpgfgraphicnamed{hopf-law}
\InputIfFileExists{hopf-law.tikz}{}{\input{./figures/hopf-law.tikz}}
\endpgfgraphicnamed\label{eq:hopfalg}\;.
  \end{equation}
\end{definition}
Again, note the similarity to Equation~\eqref{eq:antipode-hopf}:
the difference is only by a scalar factor.

\begin{definition}
A pair $(\whiteobs, \grayobs)$ of observables on the same object
$X$ is \emph{strongly complementary} iff they are coherent and: 
\begin{equation}\label{eq:bialg}
\beginpgfgraphicnamed{bialg}
\InputIfFileExists{bialg.tikz}{}{\input{./figures/bialg.tikz}}
\endpgfgraphicnamed 
\end{equation} 
\end{definition}

To expand on this definition slightly, given a pair of strongly
complementary observables, if we consider just the monoid part of one
and the comonoid part of the other then the resulting structure is, up
to a scalar factor, a bialgebra.  Note that thanks to the up-down
symmetry induced by the $\dagger$ it doesn't matter which is the
monoid and which the comonoid.  For obvious reasons, the we say that a
pair of strongly complementary observables forms a \emph{scaled
  bialgebra}, and we refer to Equation~\eqref{eq:bialg} as the
\emph{bialgebra law}.  Notice that we have not, as yet, established
any connection between complementarity
(Definition~\ref{def:complementary-obs}) and strong complementarity.
The following theorem links the two definitions.

\begin{theorem}\label{thm:SCimpliesC}  Let \whiteobs and \grayobs be
  strongly complementarity observables; then they are complementary.
  \begin{proof} 
    Let $s$ be defined by 
    \ctikzfig{mub_antipode} 
    as per Equation~\eqref{eq:mub-antipode}.  Using the bialgebra law
    we reason:
    \ctikzfig{sc_hopf_pf} 
    The last equation relies on the fact that $\whiteeta$ is classical
    for \grayobs (and $\grayeta$ for \whiteobs), and
    \eqref{eq:class-point-transpose}.
\end{proof}
\end{theorem}
As a consequence, strongly complementary observables always form a
\emph{scaled Hopf algebra}.  Note that Theorem \ref{thm:SCimpliesC}
relies on the fact that \whiteobs and \grayobs are Frobenius algebras;
it is certainly not the case that every bialgebra is a Hopf algebra.

The converse to Theorem~\ref{thm:SCimpliesC} does not hold:  it is
possible to find coherent complementary observables in \fhilb
which are not strongly complementary.  See \cite{CD2009} for a
counterexample. 

The following lemma about the antipode for a strongly complementary
pair was shown in~\cite{KissingerThesis}.

\begin{lemma}\label{lem:antipode}
  If $(\whiteobs,\grayobs)$ are strongly complementary, and have
  enough classical points then the antipode $s$ is:
  \begin{itemize}
  \item self-adjoint;
  \item a monoid homomorphism from $\whiteobs$ to itself, and
    similarly for \grayobs; and
  \item a comonoid homomorphism from $\whiteobs$ to itself, and
    similarly for \grayobs.
  \end{itemize}
\end{lemma}

\subsection{Strong Complementarity and Phase Groups} 

For  complementary observables, classical points of one observable are
always included in the phase group of the other observable, up to a
normalizing scalar. Strong complementarity strengthens this property
to inclusion as a subgroup.  Let ${\cal K}_{\!\smallgraydot}$ be the
set of classical points of $\grayobs$ multiplied by the scalar factor
$%
\beginpgfgraphicnamed{scalar}
\begin{tikzpicture}
	\begin{pgfonlayer}{nodelayer}
		\node [style=gray dot] (0) at (0, 0.12) {};
		\node [style=white dot] (1) at (0, -0.12) {};
	\end{pgfonlayer}
	\begin{pgfonlayer}{edgelayer}
		\draw [style=none] (1) to (0);
	\end{pgfonlayer}
\end{tikzpicture}}
\endpgfgraphicnamed$.

\begin{theorem}\em\label{thm:subphasegroup} 
  Let $(\whiteobs, \grayobs)$ be strongly complementary
  observables and let $\grayobs$ have finitely many classical
  points. 
  Then ${\cal
    K}_{\!\smallgraydot}$ forms a subgroup of the phase group
  $\Phi_{\!\smallwhitedot}$ of $\whiteobs$.  The converse also holds
  when $\whiteobs$ has `enough classical points'.
\end{theorem}
\begin{proof}
By strong complementarity it straightforwardly follows that, up to a
scalar, $\whitemu$ applied to two classical points of $\grayobs$
yields again a classical point of $\grayobs$:
\[
\beginpgfgraphicnamed{bialgclosedness}
\InputIfFileExists{bialgclosedness.tikz}{}{\input{./figures/bialgclosedness.tikz}}
\endpgfgraphicnamed
\]
The unit of $\Phi_{\!\smallwhitedot}$ is, up to a scalar, also a
classical point of $\grayobs$ by coherence. Hence,  ${\cal
K}_{\!\smallgraydot}$ is a submonoid of $\Phi_{\!\smallwhitedot}$ and
any finite submonoid is a subgroup.  The converse follows from:
\[
\beginpgfgraphicnamed{bialgclosedness2}
\InputIfFileExists{bialgclosedness2.tikz}{}{\input{./figures/bialgclosedness2.tikz}}
\endpgfgraphicnamed
\]
together with the `enough classical points' assumption.
\end{proof}

Recall that the exponent of a group $G$ is the maximum order of any
element of that group: $\textrm{exp}(G) = \textrm{max}\{ |g| : g \in G
\}$.

\begin{corollary}\label{cor:orderHopf}\em
  For any pair of strongly complementary observables, let $k =
  \textrm{\rm exp}({\cal K}_{\!\smallgraydot})$. Then, assuming
  $\grayobs$ has `enough classical points':
  \begin{equation}\label{eq:orderHopf}
\beginpgfgraphicnamed{orderHopf}
\InputIfFileExists{orderHopf.tikz}{}{\input{./figures/orderHopf.tikz}}
\endpgfgraphicnamed
  \end{equation}
\end{corollary}

\begin{proof}
  In a finite abelian group, the order of any element divides
  $\textrm{\rm exp}({\cal K}_{\!\smallgraydot})$. The result then
  follows by:
  \[
\beginpgfgraphicnamed{orderHopfproof1}
\InputIfFileExists{orderHopfproof1.tikz}{}{\input{./figures/orderHopfproof1.tikz}}
\endpgfgraphicnamed
  \] 
  together with the `enough classical points' assumption.
\end{proof}

\begin{proposition}\label{prop:alterdef}\em
  For a pair of strongly  complementary observables
  $\raisebox{1mm}{%
\beginpgfgraphicnamed{classicalpointaction}
\begin{tikzpicture}
	\begin{pgfonlayer}{nodelayer}
		\node [style=none] (0) at (-1, 0.15) {};
		\node [style=white dot] (1) at (-1, -0.12) {\footnotesize $\!\!i\!\!$};
		\node [style=none] (2) at (-1, -0.4) {};
	\end{pgfonlayer}
	\begin{pgfonlayer}{edgelayer}
		\draw (1) to (2.center);
		\draw (0.center) to (1);
	\end{pgfonlayer}
\end{tikzpicture}}
\endpgfgraphicnamed}$ is a
  $\grayobs$-homomorphism for all $%
\beginpgfgraphicnamed{classicalpointwhite}
\begin{tikzpicture}
	\begin{pgfonlayer}{nodelayer}
		\node [style=none] (0) at (0, 0.25) {};
		\node [style=white dot] (1) at (0, 0) {\footnotesize $\!\!i\!\!$};
	\end{pgfonlayer}
	\begin{pgfonlayer}{edgelayer}
		\draw (0.center) to (1);
	\end{pgfonlayer}
\end{tikzpicture}}
\endpgfgraphicnamed\in{\cal
  K}_{\!\smallgraydot}$. Conversely,   this property defines strong
  complementarity provided $\whitedelta$ has `enough classical points'.
\end{proposition}

\begin{proof}
Similar to the proof of Thm.~\ref{thm:subphasegroup}.
\end{proof}

\subsection{Classification of Strong Complementarity in \fhilb}\label{sec:classification}

\begin{corollary}\em\label{col:classification}
Every pair of strongly complementary observables in \fhilb is
of the following form: 
\[
\left\{\begin{array}{cl}
\graydelta   & :: \ket{g}\mapsto \ket{g}\otimes \ket{g}\vspace{1mm}\\
\grayepsilon & :: \ket{g}\mapsto 1
\end{array}\right.
\quad
\left\{\begin{array}{cl}
\whitedelta^\dagger   & :: \ket{g}\otimes \ket{h} \mapsto {1\over\sqrt{D}} \ket{g + h}\vspace{1mm} \\
\whiteepsilon^\dagger & :: 1 \mapsto \sqrt{D} \ket{0}
\end{array}\right.
\]
where $(G =\{g, h, \ldots\}, +, 0)$ is a finite Abelian
group. Conversely, each such pair is always strongly complementary.  
\end{corollary}
\begin{proof}
By Theorem \ref{thm:subphasegroup} it follows that the classical
points of one observable (here  $\grayobs$) form a group for the
multiplication of the other observable (here $\whitedelta^\dagger$),
and in \fhilb this characterises strong complementarity.   
\end{proof}

One of the longest-standing open problems in quantum information is
the characterisation of the number of pairwise complementary
observables in a Hilbert space of dimension $D$.  In all known cases
this is $D+1$, and the smallest unknown case is $D=6$.  We now show
that in the case of strong complementarity this number is always $2$
for $D\geq 2$.

\begin{theorem}\em 
  In a Hilbert space with $D\geq 2$ the largest set of pairwise
  strongly complementary observables has size at most $2$.
\end{theorem}

\begin{proof}
  Assume that both $(\whiteobs, \grayobs)$ and  $(\whiteobs,
  \blackobs)$ are strongly complementary pairs.  By coherence
  $\grayunit$ and $\unit$ must be proportional to classical points of
  $\whiteobs$. If $(\grayobs, \blackobs)$ were to be strongly
  complementary observables, it is easily shown that
  $%
\beginpgfgraphicnamed{innerprod}
\begin{tikzpicture}[dotpic]
	\begin{pgfonlayer}{nodelayer}
		\node [style=dot] (0) at (-4, 0.25) {};
		\node [style=gray dot] (1) at (-4, -0.25) {};
	\end{pgfonlayer}
	\begin{pgfonlayer}{edgelayer}
		\draw [style=short diredge] (1) to (0);
	\end{pgfonlayer}
\end{tikzpicture}}
\endpgfgraphicnamed\not= 0$ so $\grayunit$ and $\unit$ are
  proportional to the same classical point.  Hence, up to a non-zero
  scalar: 
  \[
\beginpgfgraphicnamed{two_strong_compl_pfBIS}
\InputIfFileExists{two_strong_compl_pfBIS.tikz}{}{\input{./figures/two_strong_compl_pfBIS.tikz}}
\endpgfgraphicnamed
  \]
  i.e.~the identity has rank 1, which fails for $D\geq2$. By Corollary
  \ref{col:classification} a strongly complementary pair  exists for
  any $D\geq2$.
\end{proof}


\section{Mixed states, measurements, and ``abstract probabilities''}
\label{sec:mixed-stat-meas}

For some ket $\ket\psi$ in a Hilbert space, there are (at least) four
distinct ways to represent $\ket\psi$ as a linear map.

It is possible to represent a ket $\ket\psi \in H$ as a map $\ket\psi
\colon \mathbb C \to H$, sending $1 \in \mathbb C$ to $\ket\psi$. We
call such a map a ``point'' of $H$, because it does nothing more than
picking out a specific element. The second map is the associated bra
$\bra\psi \colon H \to \mathbb C$. This kind of map is called a 
``co-point''. We can also regard such a map as an element of the dual space
$H^*$. But then, $H^*$ is just another Hilbert space, so we could just
as well represent $\bra\psi$ as a point of $H^*$. That is, a linear
map $\bra\psi^* \colon \mathbb C \to H^*$. There is yet a fourth way to
represent $\ket\psi$, namely as a linear map $\ket\psi^* \colon H^*
\to \mathbb C$, sending a bra $\bra\phi \in H^*$ to the inner product
$\braket\phi\psi \in \mathbb C$.

So, for a given ket $\ket\psi$, there are four ways to write it as
points or copoints.
\ctikzfig{four_points}

The difference in these four pictures is largely notational: the data
they represent is the same. However, its a very useful piece of
notation, especially when representing functionals between spaces of
maps. Firstly, we note that we can represent a map $M \colon A \to B$
as a special kind of point, $\ket{\Psi_M} \in A^* \otimes B$. \[ M =
\sum a_i^j \ket{j}\bra{i} \qquad \mapsto \qquad \ket{\Phi_M} = \sum
a_i^j \bra{i} \otimes \ket{j} \]

These two objects clearly represent the same data. In fact, this
mapping is essentially the Choi-Jamio\l{}kowski isomorphism. However
instead of fixing a basis, we rely on the dual space $A^*$. Thus the
value on the right does not depend on a choice of orthonormal basis. By fixing a
basis $\mathcal B = \{ \ket{i} \}$, we can define a transposition map
$T_{\mathcal{B}}(\ket{i}) = \bra{i}$. Then the usual C-J isomorphism
is recovered as $(T_{\mathcal{B}} \otimes 1)\ket{\Phi_M}$. However,
this \textit{does} depend on a choice of basis, since
$T_{\mathcal{B}}$ does.

In~\cite{AC2004}, the authors refer to $\ket{\Psi_M}$ as the ``name''
of a map. We shall often find this representation more useful than the
usual C-J representation, especially in instances involving several
distinct orthonormal bases. Using the caps and cups from before, we can
isomorphically relate maps and their associated names.
\ctikzfig{map_state_duality}

The benefit of working with names of maps, as opposed to the maps
themselves becomes clear when we start considering higher-order
functionals. For a finite-dimensional Hilbert space $H$, let $L(H)$ be
the space of linear maps from $H$ to itself. When operating on density
matrices, we often want to consider maps of the form $\Phi \colon L(H)
\to L(K)$. We can either treat this as a genuine, higher-order map, or
we can treat it as a first-order map from names to names.
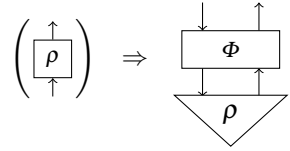
\[
\left(
\boxmap{\rho}
\right) \ \ \ \Rightarrow\ \ \ 
\beginpgfgraphicnamed{higher_order_map}
\InputIfFileExists{higher_order_map.tikz}{}{\input{./figures/higher_order_map.tikz}}
\endpgfgraphicnamed
\]

Since, in finite dimensions, we have an isomorphism $L(H) \cong H^*
\otimes H$, we know that all maps $\Phi \colon L(H) \to L(K)$ can be
represented this way.

In ordinary quantum theory, mixed quantum states are represented as
positive operators and operations as completely positive maps, or
CPMs. These are maps that take positive operators to positive
operators. A general CPM can be written in terms of a set of linear
maps $\{ B_i :  H \rightarrow K \}$ called its \textit{Kraus maps}.
\[ \Theta(\rho) = \sum_i B_i \rho B_i^\dagger \]

As before, we can collapse the higher-order map $\Theta$ to a first-
order map by translating the positive operator $\rho$ to its associated
name.
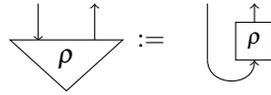
\ctikzfig{rho_point}

Then, we can encode the Kraus vectors of $\Theta$ in a map $B' = \sum
\ket{i} \otimes B_i$ and represent $\Theta$ as: 
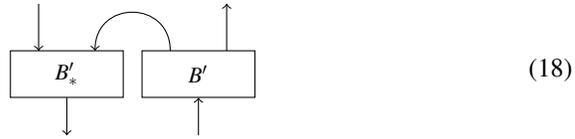
\begin{equation}
\beginpgfgraphicnamed{kraus_decomp}
\InputIfFileExists{kraus_decomp.tikz}{}{\input{./figures/kraus_decomp.tikz}}
\endpgfgraphicnamed
\end{equation}

When we take the elements in Eq. (\ref{eq:cp-map}) to be morphisms in
an arbitrary $\dagger$-compact category, this gives us an abstract
definition of a completely positive map. This is Selinger's
representation of CPMs~\cite{SelingerCPM}. 
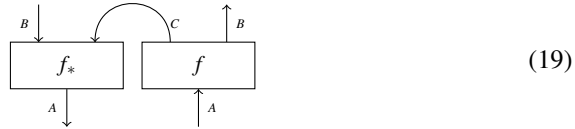
\begin{equation}\label{eq:cp-map}
\beginpgfgraphicnamed{abstract_cp}
\InputIfFileExists{abstract_cp.tikz}{}{\input{./figures/abstract_cp.tikz}}
\endpgfgraphicnamed
\end{equation}

Important special cases are \textit{states} where $A \cong I$,
\textit{effects} where $B \cong I$, and `pure' maps, where $C \cong
I$.

\subsection{Measurements and Born vectors}

Returning to quantum mechanics, we can see how a quantum measurement
would look in this language. A (projective) quantum measurement
$M_{\smallwhitedot}$ is a CPM that sends trace 1 positive operators
(in this case quantum states) to trace 1 positive operators that are
diagonal in some ONB (encoding a probability distribution of
outcomes). Suppose we wish to measure with respect to an
observable $\grayobs$,
whose classical points form an ONB $\{ \ket{x_i} \}$. The probability
of getting the $i$-th measurement outcome is computed using the Born
rule. 
\[ \textrm{Prob}(i, \rho) = \textrm{Tr}(\ketbra{x_i}{x_i} \rho) \]

We can write this probability distribution as a vector in the basis $\{ \ket{x_i} \}$. That is, a vector whose $i$-th entry is the probability of the $i$-th outcome:
\[ M_{\!\smallgraydot}(\rho) = \sum \textrm{Tr}(\ketbra{x_i}{x_i} \rho) \ket{x_i} \]

So, $M$ defines a linear map from density matrices to probability distributions. Expanding this graphically, we have:
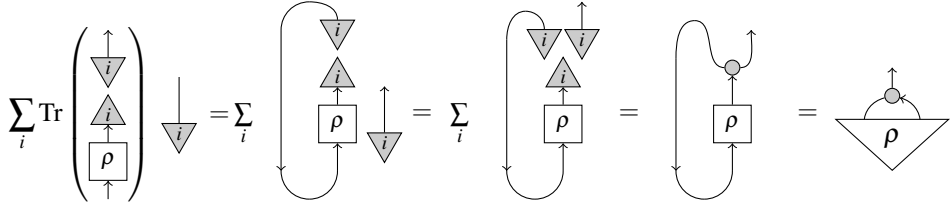
\[
\sum_i \textrm{Tr}\left( %
\beginpgfgraphicnamed{P_rho}
\InputIfFileExists{P_rho.tikz}{}{\input{./figures/P_rho.tikz}}
\endpgfgraphicnamed \right) \ \graypointmap{i}\ =
\beginpgfgraphicnamed{measurement_derivation2}
\InputIfFileExists{measurement_derivation2.tikz}{}{\input{./figures/measurement_derivation2.tikz}}
\endpgfgraphicnamed
\]

We are now ready to make definitions for abstract measurements and
abstract probability distributions, which we shall call Born vector. 

\begin{definition}
  For an observable structure $\grayobs$, a measurement is defined as
  the following map: 
  \ctikzfig{measurement}
  Any point $\roundket{\Lambda}_{\!\smallgraydot} : I \rightarrow X$ of the following
  form is called a \textit{Born vector}, with respect to $\grayobs$: 
  \ctikzfig{born_vector}
\end{definition}

\begin{theorem}\label{thm:born-vector-probs}
  In \catFHilb, Born vectors for an observable $\grayobs$ are precisely those vectors whose entries are positive and sum to $1$, when written in the basis of $\grayobs$.
\end{theorem}

\begin{proof}
  Let $\roundket{\Lambda}_{\!\smallgraydot}$ be a Born vector. Its $i$-th coefficient in the $\grayobs$-basis is given by:
  \ctikzfig{born_vector_positive_num}
  We can see that these coefficients sum to $1$ by using the definition of the deleting point:
  \ctikzfig{born_vector_sum_to_one}
  This is a completely positive map from $\mathbb C$ to $\mathbb C$. In other words, it is a positive scalar. For the converse, assume $\roundket\Lambda_{\!\smallgraydot}$ is a probability distribution whose $i$-th coefficient in the $\grayobs$-basis is $p_i \in \mathbb R_+$. Then, letting $\psi = \sum \sqrt{p_i} \graypointmap{i}$:
  \ctikzfig{born_vector_converse}
  Post-composing with the deleting point yields $\sum (\sqrt{p_i})^2 = \sum p_i = 1$.
\end{proof}

Thus Born vectors in \catFHilb correspond precisely to \textit{probability distributions} over classical points.

We can naturally extend the definition above to points of the form
$\roundket{\Lambda}_{\!\smallgraydot} : I \rightarrow X \otimes \ldots \otimes X$ by
requiring that they be Born vectors with respect to the product
Frobenius algebra $\grayobs \otimes \ldots \otimes \grayobs$. 

The adjoint of the measurement map $m_{\smallgraydot}^\dagger$ is a
preparation operation. In \catFHilb, it takes a Born vector
$\roundket{\Lambda}_{\!\smallgraydot}$ with respect to $\grayobs$ and produces a
probabilistic mixture of the (pure) outcome states of $\grayobs$ with
probabilities given by $\roundket{\Lambda}_{\!\smallgraydot}$. 

This leads to a simple classical vs.~quantum diagrammatic paradigm
that applies in all of the models we consider~\cite{CPaqPav}: \em classical
systems are encoded as a single wire and
quantum systems as a double wire\em. The same applies to operations,
and $m_{\smallgraydot}$ and $m_{\smallgraydot}^\dagger$ allow passage
between these types. 

Note that the classical data will `remember' to which observable it
relates,  cf.~the encoding $\sum_i p_{i} \ket{x_i}$.  This is
physically meaningful since, for example, when one measures position
the resulting value will carry specification of the length unit in
which it is expressed.  If one wishes to avoid interconversion of this
`classical data with memory', one could fix one observable, and
unitarily transform the quantum data before measuring. Indeed, if  
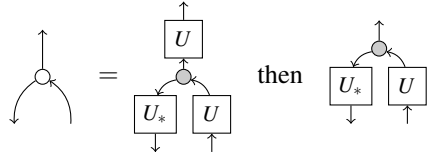
\[
\beginpgfgraphicnamed{Obs_transf}
\InputIfFileExists{Obs_transf.tikz}{}{\input{./figures/Obs_transf.tikz}}
\endpgfgraphicnamed \quad \mbox{then} \quad %
\beginpgfgraphicnamed{Obs_transf2}
\InputIfFileExists{Obs_transf2.tikz}{}{\input{./figures/Obs_transf2.tikz}}
\endpgfgraphicnamed
\]
measures the $\whiteobs$-observable but produces $\grayobs$-data. In
\catFHilb, all observable structures are unitarily isomorphic, so any
projective measurement can be obtained in this way. A particularly
relevant example is when these unitaries are phases with respect the
another observable structure $\whiteobs$.

\begin{equation}\label{eq:alphameasurement}
\graymeas^\alpha := %
\beginpgfgraphicnamed{Obs_transf3}
\InputIfFileExists{Obs_transf3.tikz}{}{\input{./figures/Obs_transf3.tikz}}
\endpgfgraphicnamed
\end{equation}

When $\whiteobs$ is induced by the Pauli spin-$Z$ observable and
$\grayobs$ by the Pauli spin-$X$ observable, then $\graymeas =
\graymeas^0$ is an $X$ measurement and $\graymeas^{\pi/2}$ is a $Y$
measurement. Note however, that both produce Born vectors of outcome
probabilities with respect to the $\grayobs$ basis. This will be
useful in the next section.


\section{Example: non-locality of QM}

In 1989 Greenburger, Horne, and Zeilinger provided an
analysis~\cite{GHZ89} of quantum theory which improves on Bell's
theorem in one crucial way. Whereas Bell demonstrated a
\textit{probabilistic} argument that quantum theory is incompatible
with the assumption of local realism (i.e. quantum theory generates
correlations that are too high for a classical local hidden variable
model), the GHZ/Mermin theorem illustrates a situation that rules out
a locally realistic model \textit{possibilistically}. That is, they
described a series of experiments for which quantum theory predicts a
single, definite outcome that is impossible under the assumption of
locality.

Here, we reproduce Mermin's version of this argument~\cite{Mermin}
using diagrammatic techniques. There are two key ingredients for this
translation. The first is a graphical notion of locality.  For our
purposes, it will suffice to treat locality as the fact that global
probability can be traced down to hidden states that determine all
measurement outcomes, since we shall show that no hidden state can
ever be compatible with the predictions of quantum theory.  Hence,
there is no point in even considering crafting a local hidden variable
representation. 

The second key ingredient is \textit{parity}. The GHZ/Mermin trick for
producing definite outcomes is to look not at individual measurement
outcomes, but the overall parity of those outcomes, i.e. ``Did the
experiment produce and even or an odd number of clicks?''. Considering
a single outcome (click or no-click) as an element of the abelian
group $\mathbb Z_2$, the parity of a set of outcomes is given by their
sum in the group. We already saw in section~\ref{sec:classification} that
strongly complementary observables are classified by abelian
groups. In two dimensions, there is only one such strongly
complementary pair, namely the one corresponding to $\mathbb
Z_2$. When we prepare a GHZ state with respect to a certain observable
(e.g. spin-$Z$) and conduct measurements using a strongly
complementary observable (e.g. spin-$X$), we will see this $\mathbb
Z_2$ structure arise.

By combining these two elements (the topological picture of
locality and the encoding of abelian groups as strongly complementary
observables) we will derive a contradiction. Furthermore in
section~\ref{sec:necessity}, we shall see how strong complementary was
embedded in the pre-conditions of a GHZ/Mermin-style argument in the
first place.




\subsection{A local hidden variable model}

For a particular $n$-party  state $|\Psi\rangle$ in some theory,
a \em local hidden variable (LHV) \em  model  for that state consists
of: 
\begin{itemize}
\item a family of hidden states $|\lambda\rangle$, each of  which
  assigns to any measurement on each  subsystem a definite outcome; and,  
\item a probability distribution on these hidden states,
\end{itemize}
which simulates the probabilities of that theory. We say that a theory
is \em local \em if each state admits a LHV model.

Consider three systems and four possible (compound) measurement
settings, $XXX$, $XYY$, $YXY$, and $YYX$.  A hidden state of an
underlying LHV model stores one measurement outcome for each setting
on each system:
\[ \ket{\lambda'} =  |\ 
  \underbrace{\overbrace{+}^X\overbrace{-}^Y}_{\textit{\footnotesize system~1}}\  
  \underbrace{\overbrace{-}^X\overbrace{+}^Y}_{\textit{\footnotesize system~2}}\ 
  \underbrace{\overbrace{-}^X\overbrace{+}^Y}_{\textit{\footnotesize system~3}}\
 \rangle \]
 We can represent this diagrammatically as follows:
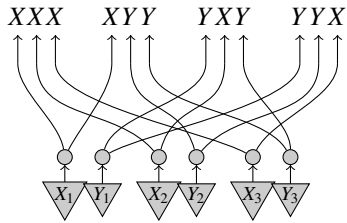
\[
\beginpgfgraphicnamed{local_hidden_stateNEW}
\InputIfFileExists{local_hidden_stateNEW.tikz}{}{\input{./figures/local_hidden_stateNEW.tikz}}
\endpgfgraphicnamed
\]
that is, we simply copy those values to each of the four measurement settings.

\subsection{Encoding the GHZ state and computing correlations, diagrammatically}\label{sec:corcomp}

To present Mermin/GHZ style argument graphically, we first show how to
compute measurement outcomes for an $n$-party GHZ state
graphically. This computation relies on a standard theorem about
bialgebras, which relates a graph-theoretic property of diagrams to
equality of bialgebra expressions.

\begin{definition}
  Let $(\graymult, \grayunit, \whitecomult, \whitecounit)$ be a commutative, cocommutative bialgebra, and let $\mathcal D$ be a diagram consisting only of $\graymult, \grayunit, \whitecomult, \whitecounit$, identity maps, and swaps. Then, the characteristic matrix $\chi$ of $\mathcal D$ is a matrix of natural numbers where the $(i,j)$-th entry represents the number of forward-directed paths connecting the $i$-th input to the $j$-th output.
\end{definition}

\begin{example}
  The following terms have characteristic matrix
  \[ \left(\begin{matrix} 0 & 1 \\ 0 & 0 \end{matrix}\right): \qquad
\beginpgfgraphicnamed{char_0100}
\InputIfFileExists{char_0100.tikz}{}{\input{./figures/char_0100.tikz}}
\endpgfgraphicnamed \]
  
  The following terms have characteristic matrix
  \[ \left(\begin{matrix} 1 & 1 \\ 1 & 1 \\ 1 & 1 \end{matrix}\right): \qquad
\beginpgfgraphicnamed{char_111111}
\InputIfFileExists{char_111111.tikz}{}{\input{./figures/char_111111.tikz}}
\endpgfgraphicnamed \]
\end{example}

\begin{theorem}\label{thm:bialg-nf}
  If two diagrams generated by the same bialgebra have the same characteristic matrix, they are equal as maps.
\end{theorem}

\begin{proof}
  (sketch) It is possible to show by case analysis that the three bialgebra equations can be used to move all of the gray dots to the top all the white dots to the bottom.
  \ctikzfig{bialg_all3}
  We can furthermore show that all three of these transformations preserve the characteristic matrix of $\mathcal D$. Once this is done, we obtain a diagram in normal form:
  \[
\beginpgfgraphicnamed{bialg_nf_example}
\InputIfFileExists{bialg_nf_example.tikz}{}{\input{./figures/bialg_nf_example.tikz}}
\endpgfgraphicnamed
    \qquad \leftrightarrow \qquad
    \left(
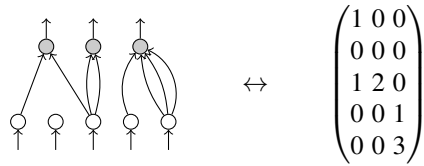
    \begin{matrix}
      1 & 0 & 0 \\
      0 & 0 & 0 \\
      1 & 2 & 0 \\
      0 & 0 & 1 \\
      0 & 0 & 3
    \end{matrix}
    \right)
  \]
  Then, it is possible to show there is \textit{exactly one} such
  normal form for each characteristic matrix. In fact, it is
  straightforward to read off the matrix by counting edges in the
  normal form. Since every diagram can be put into normal form using
  equations that preserve the characteristic matrix, and normal forms
  are in 1-to-1 correspondence with characteristic matrices, this
  completes the proof.\footnote{For a formal statement and proof of this theorem, in terms of factorisation systems see~\cite{Lack2004}.}
\end{proof}


We can now apply the theorem to prove the following corollary.

\begin{corollary}\label{cor:genbialg}
  The following equation holds for any connected bipartite graph with
  directions as shown.
  \begin{equation}\label{eq:bialgarrows}
\beginpgfgraphicnamed{directed_bialg_cor}
\InputIfFileExists{directed_bialg_cor.tikz}{}{\input{./figures/directed_bialg_cor.tikz}}
\endpgfgraphicnamed
  \end{equation}
\end{corollary}

\begin{proof}
  For the theorem on bialgebras to apply, all of the edges need to be
  directed upward. For a strongly complementary observable, the edge
  direction between two different colours can be reversed by applying
  the antipode $S$. Then, we use the fact that $S$ is a Frobenius
  algebra endomorphism to move it down.  \ctikzfig{directed_bialg_pf1}

  We apply Theorem \ref{thm:bialg-nf} and the spider theorem to complete the proof.
  \ctikzfig{directed_bialg_pf2}
\end{proof}

We compute the classical probability distributions (= $\grayobs$-data)
for $n$ measurements against arbitrary phases
$\alpha_i\in \Phi_{\!\smallwhitedot}$ on $n$ systems of any type in a 
generalised $GHZ^n_{\!\smallwhitedot\!}$-state:
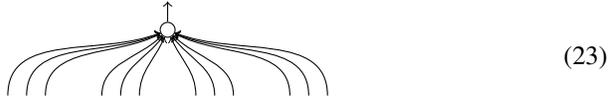
\[ %
\beginpgfgraphicnamed{CorrelationComp1}
\InputIfFileExists{CorrelationComp1.tikz}{}{\input{./figures/CorrelationComp1.tikz}}
\endpgfgraphicnamed
   \stackrel{(\ref{eq:decspidercomp})}{=}
\beginpgfgraphicnamed{CorrelationComp2}
\InputIfFileExists{CorrelationComp2.tikz}{}{\input{./figures/CorrelationComp2.tikz}}
\endpgfgraphicnamed = (*) \]

Applying Corollary \ref{cor:genbialg}, we note that this is a
probability distribution followed by a $\whitedot$-copy.
\begin{equation}\label{eq:correlationspider}
(*) = %
\beginpgfgraphicnamed{CorrelationComp3}
\InputIfFileExists{CorrelationComp3.tikz}{}{\input{./figures/CorrelationComp3.tikz}}
\endpgfgraphicnamed =: %
\beginpgfgraphicnamed{angle_dist}
\InputIfFileExists{angle_dist.tikz}{}{\input{./figures/angle_dist.tikz}}
\endpgfgraphicnamed 
\end{equation}

The following is an immediate consequence.
\begin{theorem}\label{thm:symmetriccorr}\em
When measuring each system of a $GHZ^n_A$-state against an arbitrary
angle then the resulting classical probability distribution over
outcomes is symmetric.  
\end{theorem}

\begin{theorem}\em
The classical probability distributions for
$\graymeas^{\alpha_1}\otimes\ldots\otimes
\graymeas^{\alpha_n}$-measurements on a $GHZ^n_A$-state is: 
\begin{itemize}
\item uncorrelated if $\roundket{\sum \alpha_i}_{\!\smallgraydot}$ is a classical point
for $\whiteobs$ and,
\item parity-correlated if $\roundket{\sum \alpha_i}_{\!\smallgraydot}$ is a classical
point $i$ for $\grayobs$ (i.e. contains precisely those outcomes
$i_1 \otimes \ldots \otimes i_n$ such that the sum of group elements
$\sum i_k$ is equal to $i$).
\end{itemize}
\end{theorem}

\begin{example}
We can compute the outcome
distributions for $XXX$,  $XYY$, $YXY$, and $YYX$ measurements on
three qubits in a GHZ-state using the technique described above.
First, outcome distribution $\roundket{A}_{\!\smallgraydot}$ for $XXX$:
\[
\beginpgfgraphicnamed{XXX_corrs}
\InputIfFileExists{XXX_corrs.tikz}{}{\input{./figures/XXX_corrs.tikz}}
\endpgfgraphicnamed
\]
Next, we compute outcome distribution $\roundket{B_1}_{\!\smallgraydot}$ for $XYY$:
\[
\beginpgfgraphicnamed{XYY_corrs}
\InputIfFileExists{XYY_corrs.tikz}{}{\input{./figures/XYY_corrs.tikz}}
\endpgfgraphicnamed
\]

Computing correlations as in Figure~(\ref{eq:correlationspider}) is symmetric in the choice of measurement angle for each of the systems. Thus, for the other two cases ($YXY$ and $YYX$), we get the same result:
$\roundket{B_1}_{\!\smallgraydot} = \roundket{B_2}_{\!\smallgraydot} = \roundket{B_3}_{\!\smallgraydot}$.
\end{example}

\subsection{Deriving the contradiction}\label{sec:contra}
Consider the function:
\begin{equation}\label{eq:parity-mapping}
\beginpgfgraphicnamed{constant_functionNEW}
\InputIfFileExists{constant_functionNEW.tikz}{}{\input{./figures/constant_functionNEW.tikz}}
\endpgfgraphicnamed
\end{equation}
\noindent
We have already seen that strongly complementary observables correspond to group algebras. That is, $\whitemult$ defines a group algebra over the classical points of $\grayobs$. For qubits there is only one choice: $\mathbb Z_2$. Thus, this function computes the parity (i.e. the $\mathbb Z_2$-sum) of
 all  outcomes.  

Measuring the parity for any local hidden state we obtain:
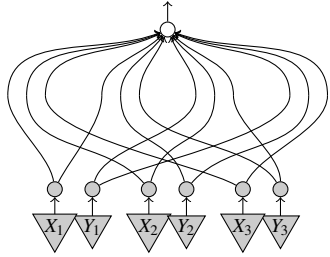
\[
\beginpgfgraphicnamed{NEW1}
\InputIfFileExists{NEW1.tikz}{}{\input{./figures/NEW1.tikz}}
\endpgfgraphicnamed
\]
that is, by (\ref{eq:orderHopf}):
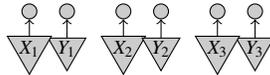
\[
\beginpgfgraphicnamed{NEW2}
\InputIfFileExists{NEW2.tikz}{}{\input{./figures/NEW2.tikz}}
\endpgfgraphicnamed
\]
and hence:  
\[
\beginpgfgraphicnamed{NEW3}
\begin{tikzpicture}[dotpic]
	\begin{pgfonlayer}{nodelayer}
		\node [style=white dot] (0) at (0, 1.75) {};
		\node [style=none] (1) at (0, 2.5) {};
	\end{pgfonlayer}
	\begin{pgfonlayer}{edgelayer}
		\draw [style=diredge] (0) to (1.center);
	\end{pgfonlayer}
\end{tikzpicture}}
\endpgfgraphicnamed
\]
and measuring the parity in quantum theory we obtain: 
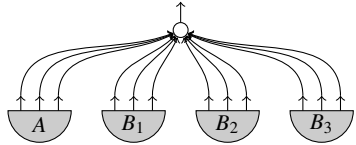
\[
\beginpgfgraphicnamed{NEW4}
\InputIfFileExists{NEW4.tikz}{}{\input{./figures/NEW4.tikz}}
\endpgfgraphicnamed
\]
that is, by the previous section: 
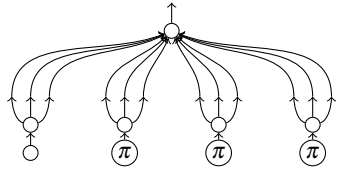
\[
\beginpgfgraphicnamed{NEW5}
\InputIfFileExists{NEW5.tikz}{}{\input{./figures/NEW5.tikz}}
\endpgfgraphicnamed
\]
and hence:
\[
\beginpgfgraphicnamed{NEW6}
\begin{tikzpicture}[dotpic]
	\begin{pgfonlayer}{nodelayer}
		\node [style=white dot, inner sep=1 pt, font={\footnotesize}] (0) at (0, 1.75) {$\pi$};
		\node [style=none] (1) at (0, 2.5) {};
	\end{pgfonlayer}
	\begin{pgfonlayer}{edgelayer}
		\draw [style=diredge] (0) to (1.center);
	\end{pgfonlayer}
\end{tikzpicture}}
\endpgfgraphicnamed
\]
which yields a contradiction.

\subsection{GHZ/Mermin assumptions and the necessity of strong complementarity}\label{sec:necessity}

We shall examine two assumptions that play a key role in a GHZ/Mermin style non-locality argument, and show that the presence of a strongly complementary observable arises as a consequence.

The original argument due to Greenburger, Horne, and Zeilinger~\cite{GHZ89} and later simplifications~\cite{GHZ,Mermin} focus on a state defined in terms of correlated (or anti-correlated) $Z$-spins and local spin measurements in the XY-plane. We generalise this assumption as follows.

\begin{description}
  \item[\textbf{Assumption 1 (Coherence).}] We will use a GHZ state defined with respect to some observable structure $\whiteobs$. Measurements are all conducted within a $\whiteobs$-phase of some coherent observable $\grayobs$.
\end{description}

In \catFHilb, all observables containing at least one unbiased classical point, w.r.t. $\whiteobs$, are within a $\whiteobs$-phase of a coherent observable, so we could weaken this assumption further. That is, if $\grayobs$ contains an unbiased classical point, we might as well assume it is coherent, since \textbf{Assumption 1} allows us to freely choose phases.

By \textbf{Assumption 1}, the correlations associated with each experiment are computed from this diagram:
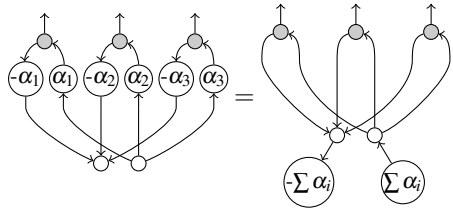
\[ %
\beginpgfgraphicnamed{ghz3}
\InputIfFileExists{ghz3.tikz}{}{\input{./figures/ghz3.tikz}}
\endpgfgraphicnamed = %
\beginpgfgraphicnamed{ghz3_simplified}
\InputIfFileExists{ghz3_simplified.tikz}{}{\input{./figures/ghz3_simplified.tikz}}
\endpgfgraphicnamed \]

The second assumption is what~\cite{GHZ89} refers to as ``super-classicality''. We shall refer to it as \textit{sharpness}.

\begin{description}
\item[\textbf{Assumption 2 (Sharpness).}] After all subsystems except
  one are measured, the final measurement outcome is fixed.
\end{description}

The map %
\beginpgfgraphicnamed{decoh}
\InputIfFileExists{decoh.tikz}{}{\input{./figures/decoh.tikz}}
\endpgfgraphicnamed is called the \textit{decoherence map} for $\grayobs$. It projects from the space of all quantum mixed states to the the space of classical mixtures of eigenstates of $\grayobs$. To assert sharpness, we require that, once two of the three systems are measured, the third is invariant under this map:
\begin{equation}\label{eq:ConverseProof1}
\beginpgfgraphicnamed{ConverseProof1}
\InputIfFileExists{ConverseProof1.tikz}{}{\input{./figures/ConverseProof1.tikz}}
\endpgfgraphicnamed
\end{equation}
Plugging the unit of $\whiteobs$ in the 2nd system both for LHS and RHS, and using coherence we obtain:
\begin{equation}\label{eq:ConverseProof2}
\beginpgfgraphicnamed{ConverseProof2}
\InputIfFileExists{ConverseProof2.tikz}{}{\input{./figures/ConverseProof2.tikz}}
\endpgfgraphicnamed 
\end{equation}
and  by exploiting symmetry we have: 
\begin{equation}\label{eq:ConverseProof3}
\beginpgfgraphicnamed{ConverseProof3}
\InputIfFileExists{ConverseProof3.tikz}{}{\input{./figures/ConverseProof3.tikz}}
\endpgfgraphicnamed
\end{equation}
Hence we obtain:
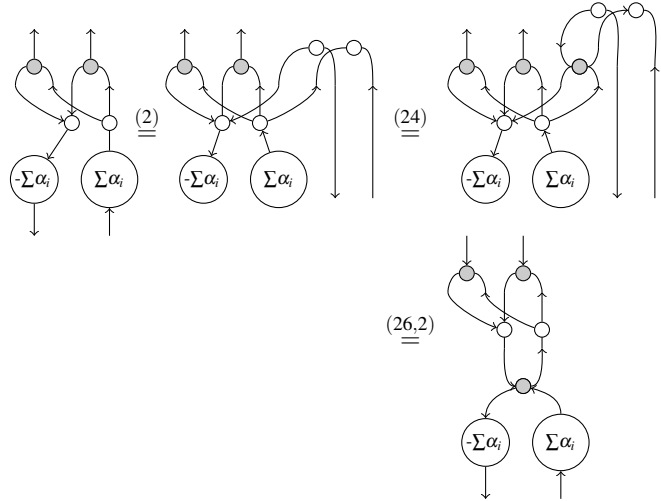
\[
\beginpgfgraphicnamed{ConverseProof4}
\InputIfFileExists{ConverseProof4.tikz}{}{\input{./figures/ConverseProof4.tikz}}
\endpgfgraphicnamed
\]
Since $\whitedelta^\dagger\circ(1_X\otimes \sum_i\alpha_i)$ is unitary
it cancels. Thus our assumptions lead us to conclude the following
equation for the observable structures $(\whiteobs, \grayobs)$:
\begin{equation}\label{eq:bialgalt}
\beginpgfgraphicnamed{bialgaltbis}
\InputIfFileExists{bialgaltbis.tikz}{}{\input{./figures/bialgaltbis.tikz}}
\endpgfgraphicnamed
\end{equation}
 
\begin{proposition}\label{prop:bialgalt}\em
A pair $(\whiteobs, \grayobs)$ of coherent observables satisfying Equation (\ref{eq:bialgalt}) are strongly complementary.
\end{proposition}

\begin{proof}
  First, we show that equation (\ref{eq:bialgalt}) implies the following, for any pair of coherent observables:
  \begin{equation}\label{eq:bialg-funkydirs}
\beginpgfgraphicnamed{bialg_funkydirs}
\InputIfFileExists{bialg_funkydirs.tikz}{}{\input{./figures/bialg_funkydirs.tikz}}
\endpgfgraphicnamed
  \end{equation}

  The proof goes as follows:
  \ctikzfig{bialg_funkydirs_pf1}
  ...which implies:
  \ctikzfig{bialg_funkydirs_pf2}

  Equation (\ref{eq:bialg-funkydirs}) is very nearly the required equation for strong complementarity, but the directions are wrong. However, we can correct this by first showing the following equations, using coherence and (\ref{eq:bialg-funkydirs}):
  \ctikzfig{bialg_dot_copy_pf1}
  \ctikzfig{bialg_dot_copy_pf2}
  
  Then, we complete the proof by using the equations above to change the directions of the arrows on the inside:
  \ctikzfig{bialg_alt_form_pf1}
  \ctikzfig{bialg_alt_form_pf2}
  Thus any coherent pair of observable structures satisfying Equation (\ref{eq:bialgalt}) is a strongly complementary pair.
\end{proof}


\section{Summary and Further Reading}


In this chapter, we developed the notion of a \emph{generalised
compositional theory}, a new approach to studying quantum mechanics and
constructing foil theories with quantum-like properties. The main building
blocks for a GCT are:
\begin{itemize}
  \item a collection of systems $A, B, C, \ldots$,
  \item a collection of \textit{primitive} processes, and
  \item a notion of horizontal composition $\otimes$ and vertical composition $\circ$.
\end{itemize}
From this sparse setting, we began to add extra pieces of structure.
\begin{itemize}
  \item symmetry maps $\Rightarrow$ ``permutibility of systems''
  \item dagger $\Rightarrow$ ``time-reversed processes''
  \item duals $\Rightarrow$ ``map/state duality''
\end{itemize}

This structure and its diagrammatic presentation give a rich language for talking about composed processes. We then went on to define various concepts in this framework, often by analogy to their Hilbert space counterparts: pure states, reversible dynamics, quantum observables, complementarity, mixed states, and measurements. Using these ingredients, we worked through a complete example, following Mermin's illustration of a possibilistic locality violation, as predicted by quantum mechanics.

The interested reader can find many papers related to, or extending the formalism introduced in this chapter. One example is the \textit{ZX-calculus}, which is a graphical calculus for the interaction of the Pauli-Z and Pauli-X observable structures. In addition to the usual rules (complementarity, strong complementarity), several other rules are added, which turn out to be complete for stabiliser quantum mechanics~\cite{BackensCompleteness}. This calculus has been applied to the study of measurement-based quantum computing~\cite{DuncanPerdrix2010}, topological MBQC~\cite{Horsman}, and quantum protocols~\cite{Anne}.

The ideas developed in section~\ref{sec:mixed-stat-meas} originated in~\cite{CPaqPav}. In~\cite{CPstar}, a simplified formalism for interacting classical and quantum data was developed, and can be viewed as an abstraction of the C*-algebraic approach to the study of quantum information.


\bibliographystyle{spphys}
\bibliography{merminchapter}

\begin{thebibliography}{10}
\providecommand{\url}[1]{{#1}}
\providecommand{\urlprefix}{URL }
\expandafter\ifx\csname urlstyle\endcsname\relax
  \providecommand{\doi}[1]{DOI \discretionary{}{}{}#1}\else
  \providecommand{\doi}{DOI \discretionary{}{}{}\begingroup
  \urlstyle{rm}\Url}\fi

\bibitem{Barrett2007}
J.~Barrett, Physical Review A \textbf{75}(032304) (2007)

\bibitem{InfoCaus}
M.~Pawlowski, T.~Paterek, D.~Kazlikowski, V.~Scarani, A.~Winter, M.~Zukowski,
  Nature \textbf{461}(1101) (2009).
\newblock {a}rXiv:0905.2292

\bibitem{InfoCaus2}
H.~Barnum, J.~Barrett, L.O. Clark, M.~Leifer, R.W. Spekkens, N.~Stepanik,
  A.~Wilce, R.~Wilke, New Journal of Physics \textbf{12}(033024) (2009).
\newblock {a}rXiv:0909.5075

\bibitem{Schrodinger}
E.~Schr{\"o}dinger, in \emph{Proceedings of the Cambridge Philosophical
  Society}, vol.~31 (Academic Press, 1935), pp. 555--563

\bibitem{Penrose}
R.~Penrose, in \emph{Combinatorial Mathematics and its Applications} (Academic
  Press, 1971), pp. 221--244

\bibitem{KellyLaplaza}
G.M. Kelly, M.L. Laplaza, Journal of Pure and Applied Algebra \textbf{19}, 193
  (1980)

\bibitem{JS:1993:GeoTenCal1}
A.~Joyal, R.~Street, Advances in Mathematics \textbf{102}, 20 (1993)

\bibitem{AC2004}
S.~Abramsky, B.~Coecke, in \emph{Proceedings of 19th IEEE conference on Logic
  in computer science} (IEEE Press, 2004), LiCS'04, pp. 415--425

\bibitem{CoeckeKinder}
B.~Coecke, in \emph{Quantum Theory: Reconsiderations of the Foundations III}
  (AIP Press, 2005), pp. 81--98

\bibitem{Chiri}
G.M. D'Ariano, {G. Chiribella}, P.~Perinotti, Physical Review A
  \textbf{84}(012311) (2010).
\newblock {a}rXiv:1011.6451

\bibitem{Hardy}
L.~Hardy, in \emph{Deep Beauty: Understanding the Quantum World through
  Mathematical Innovation} (Cambridge University Press, 2011), pp. 409--442.
\newblock {a}rXiv:0912.4740

\bibitem{CES2011}
B.~Coecke, B.~Edwards, R.W. Spekkens, Electronic Notes in Theoretical Computer
  Science \textbf{270}(2), 15 (2011)

\bibitem{Bill2}
B.~Edwards, {Phase groups and local hidden variables}.
\newblock Tech. Rep. RR-10-15, Dept. of computer science, University of Oxford
  (2010)

\bibitem{CoeckeKissinger2010}
B.~Coecke, A.~Kissinger, in \emph{Lecture Notes in computer science}, vol. 6199
  (Springer, 2010), pp. 297--308

\bibitem{DuncanPerdrix2010}
R.~Duncan, S.~Perdrix, in \emph{Proceedings of the 37th international
  colloquium conference on Automata, languages and programming: Part II}
  (Springer-Verlag, Berlin, Heidelberg, 2010), ICALP'10, pp. 285--296.
\newblock \urlprefix\url{http://dl.acm.org/citation.cfm?id=1880999.1881030}

\bibitem{Horsman}
C.~Horsman, New Journal of Physics \textbf{13}(095011) (2011).
\newblock {a}rXiv:1101.4722

\bibitem{MacLane:CatsWM:1971}
S.~Mac~Lane, \emph{Categories for the Working Mathematician (2nd Ed.)}
  (Springer-Verlag, 1997)

\bibitem{BobEric2011cats}
B.~Coecke, E.O. Paquette, in \emph{New Structures for Physics}, \emph{Springer
  Lecture Notes in Physics}, vol. 813 (2011), pp. 173--286

\bibitem{SelingerSurvey}
P.~Selinger, in \emph{New Structures for Physics}, \emph{Springer Lecture Notes
  in Physics}, vol. 813 (2011), pp. 289--355

\bibitem{Lafont:2003qy}
Y.~Lafont, Journal of Pure and Applied Algebra \textbf{184}(2-3), 257 (2003)

\bibitem{Adriano-Barenco:1995qy}
A.~Barenco, C.H. Bennett, R.~Cleve, D.P. DiVincenzo, N.~Margolus, P.~Shor,
  T.~Sleator, J.A. Smolin, H.~Weinfurter, Phys. Rev. A \textbf{52}, 3457
  (1995).
\newblock \doi{10.1103/PhysRevA.52.3457}

\bibitem{CPaqPav}
B.~Coecke, E.O. Paquette, D.~Pavlovic, in \emph{Semantic Techniques for Quantum
  Computation} (Cambridge University Press, 2009), pp. 29--69

\bibitem{Dieks}
D.G.B.J. Dieks, Physics Letters A \textbf{92}, 271 (1982)

\bibitem{WZ}
W.K. Wootters, W.~Zurek, Nature \textbf{299}, 802 (1982)

\bibitem{Pati}
A.K. Pati, S.L. Braunstein, Nature \textbf{404}, 164 (2000).
\newblock {a}rXiv:quant-ph/9911090

\bibitem{CPV2008}
B.~Coecke, D.~Pavlovic, J.~Vicary, Mathematical Structures in Computer Science
  \textbf{23}, 555 (2013)

\bibitem{Dusko}
D.~Pavlovic, in \emph{Lecture Notes in computer science}, vol. 5494 (Springer,
  2009), pp. 143--157.
\newblock {a}rXiv:0812.2266

\bibitem{CD2008}
B.~Coecke, R.~Duncan, in \emph{Lecture Notes in computer science}, vol. 5126
  (Springer, 2008), pp. 298--310

\bibitem{CD2009}
B.~Coecke, R.~Duncan, New Journal of Physics \textbf{13}(043016) (2011).
\newblock {a}rXiv:0906.4725

\bibitem{Spekkens}
R.W. Spekkens, Physical Review A \textbf{75}(032110) (2007).
\newblock {a}rXiv:quant-ph/0401052

\bibitem{EdwardsSpek2011}
B.~Coecke, B.~Edwards.
\newblock {Spekkens's toy theory as a category of processes}.
\newblock {a}rXiv:1108.1978v1 [quant-ph] (2011)

\bibitem{KissingerThesis}
A.~Kissinger, Pictures of processes: Automated graph rewriting for monoidal
  categories and applications to quantum computing.
\newblock Ph.D. thesis, University of Oxford (2012)

\bibitem{SelingerCPM}
P.~Selinger, Electronic Notes in Theoretical Computer Science \textbf{170}, 139
  (2007)

\bibitem{GHZ89}
D.M. Greenberger, M.A. Horne, A.~Zeilinger, in \emph{Bell's Theorem, Quantum
  Theory, and Conceptions of the Universe}, ed. by M.~Kafatos (Springer, 1989),
  pp. 69--72

\bibitem{Mermin}
N.D. Mermin, American Journal of Physics \textbf{58}, 731 (1990)

\bibitem{Lack2004}
S.~Lack, Theory and Applications of Categories \textbf{13}(9), 147 (2004)

\bibitem{GHZ}
M.A. Horne, A.~Shimony, {D M Greenberger}, A.~Zeilinger, American Journal of
  Physics \textbf{58}, 1131 (1990)

\bibitem{BackensCompleteness}
M.~Backens, in \emph{{Proceedings of Quantum Physics and Logic}} (2012), pp.
  15--27

\bibitem{Anne}
A.~Hillebrand, {Quantum protocols involving multiparticle entanglement and
  their representations in the zx-calculus}.
\newblock Master's thesis, University of Oxford (2011)

\bibitem{CPstar}
B.~Coecke, C.~Heunen, A.~Kissinger, in \emph{{Proceedings of Quantum Physics
  and Logic}} (2012), pp. 87--100

\end{thebibliography}

\end{document}